\documentclass[aps,preprintnumbers,nofootinbib,superscriptaddress,10pt]{revtex4-1}

\usepackage{amsfonts,amssymb,amsmath,amsthm,graphicx}
\usepackage{url}
\usepackage{dsfont}
\usepackage{bbm}

\theoremstyle{definition}
\newtheorem{theorem}{Theorem}[section]
\newtheorem{proposition}[theorem]{Proposition}

\newtheorem{lemma}[theorem]{Lemma}
\newtheorem{definition}[theorem]{Definition}
\newtheorem{corollary}[theorem]{Corollary}
\newtheorem{remark}[theorem]{Remark}

\newcommand*{\cB}{\mathcal{B}}

\newcommand*{\cE}{\mathcal{E}}
\newcommand*{\cF}{\mathcal{F}}

\newcommand*{\cH}{\mathcal{H}}
\newcommand*{\cI}{\mathcal{I}}

\newcommand*{\cL}{\mathcal{L}}

\newcommand*{\cP}{\mathcal{P}}
\newcommand*{\cS}{\mathcal{S}}

\newcommand*{\cT}{\mathcal{T}}

\newcommand*{\tr}{\mathrm{tr}}
\newcommand*{\ket}[1]{| #1 \rangle}
\newcommand*{\bra}[1]{\langle #1 |}
\newcommand*{\spr}[2]{\langle #1 | #2 \rangle}
\newcommand*{\braket}[2]{\langle #1 | #2 \rangle}
\newcommand*{\proj}[1]{\ket{#1}\bra{#1}}

\newcommand*{\genFid}{\bar{F}}
\newcommand*{\eps}{\varepsilon}

\newcommand*{\1}{\mathbbm{1}}

\newcommand{\mycomment}[1]{}

\begin{document}

\title{The Quantum Reverse Shannon Theorem based on One-Shot Information Theory}

\author{Mario \surname{Berta}}
\email[]{berta@phys.ethz.ch}
\affiliation{Institute for Theoretical Physics, ETH Zurich, 8093 Zurich, Switzerland.}
\author{Matthias \surname{Christandl}}
\email[]{christandl@phys.ethz.ch}
\affiliation{Institute for Theoretical Physics, ETH Zurich, 8093 Zurich, Switzerland.}
\author{Renato \surname{Renner}}
\email[]{renner@phys.ethz.ch}
\affiliation{Institute for Theoretical Physics, ETH Zurich, 8093 Zurich, Switzerland.}

\date{\today}

\begin{abstract}
The Quantum Reverse Shannon Theorem states that any quantum channel can be simulated by an unlimited amount of shared entanglement and an amount of classical communication equal to the channel's entanglement assisted classical capacity. In this paper, we provide a new proof of this theorem, which has previously been proved by Bennett, Devetak, Harrow, Shor, and Winter. Our proof has a clear structure being based on two recent information-theoretic results: one-shot Quantum State Merging and the Post-Selection Technique for quantum channels.
\end{abstract}

\maketitle

\section{Introduction}

The birth of classical information theory can be dated to 1948, when Shannon derived his famous \textit{Noisy Channel Coding Theorem}~\cite{Shannon48}. It shows that the capacity $C$ of a classical channel $\cE$ is given by the maximum, over the input distributions $X$, of the mutual information between the input $X$ and the output $\cE(X)$. That is
\begin{align*}
C(\cE)=\max_{X}\left\{H(X)+H(\cE(X))-H(X,\cE(X))\right\}\ ,
\end{align*}
where $H$ denotes the Shannon entropy. Shannon also showed that the capacity does not increase if one allows to use shared randomness between the sender and the receiver. In 2001 Bennett et al.~\cite{Bennett02} proved the \textit{Classical Reverse Shannon Theorem} which states that, given free shared randomness between the sender and the receiver, every channel can be simulated using an amount of classical communication equal to the capacity of the channel. This is particularly interesting because it implies that in the presence of free shared randomness, the capacity of a channel $\cE$ to simulate another channel $\cF$ is given by the ratio of their plain capacities $C_{R}(\cE,\cF)=\frac{C(\cE)}{C(\cF)}$ and hence only a single parameter remains to characterize classical channels.\\

\noindent
In contrast to the classical case, a quantum channel has various distinct capacities~\cite{Holevo98,Schumacher97,Bennett02,Devetak_05,Lloyd97,Shor_02}. In~\cite{Bennett02} Bennett et al.~argue that the entanglement assisted classical capacity $C_{E}$ of a quantum channel $\cE$ is the natural quantum generalization of the classical capacity of a classical channel. They show that the entanglement assisted classical capacity is given by the quantum mutual information
\begin{align*}
C_{E}(\cE)=\max_{\rho}\left\{H(\rho)+H(\cE(\rho))-H((\cE\otimes\cI)\Phi_{\rho})\right\}\ ,
\end{align*}
where the maximum ranges over all input distributions $\rho$, $\Phi_{\rho}$ is a purification of $\rho$, $\cI$ is the identity channel, and $H$ denotes the von Neumann entropy. Motivated by this, they conjectured the \textit{Quantum Reverse Shannon Theorem (QRST)} in~\cite{Bennett02}. Subsequently Bennett, Devetak, Harrow, Shor and Winter proved the theorem in~\cite{Bennett06}. The theorem states that any quantum channel can be simulated by an unlimited amount of shared entanglement and an amount of classical communication equal to the channel's entanglement assisted classical capacity. So if entanglement is for free we can conclude, in complete analogy with the classical case, that the capacity of a quantum channel $\cE$ to simulate another quantum channel $\cF$ is given by $C_{E}(\cE,\cF)=\frac{C_{E}(\cE)}{C_{E}(\cF)}$ and hence only a single parameter remains to characterize quantum channels.\\

\noindent
In addition, again analogue to the classical scenario~\cite{winter_new,Bennett06}, the Quantum Reverse Shannon Theorem gives rise to a \textit{strong converse} for the entanglement assisted classical capacity of quantum channels. That is, if one sends classical information through a quantum channel $\cE$ at a rate of $C_{E}(\cE)+\varsigma$ for some $\varsigma>0$ (using arbitrary entanglement as assistance), then the fidelity of the coding scheme decreases exponentially in $\varsigma$~\cite{Bennett06}.\\

\noindent
Free entanglement in quantum information theory is usually given in the form of maximally entangled states. But for the Quantum Reverse Shannon Theorem it surprisingly turned out that maximally entangled states are not the appropriate resource for general input sources. More precisely, even if one has arbitrarily many maximally entangled states as an entanglement resource, the Quantum Reverse Shannon Theorem cannot be proven~\cite{Bennett06}. This is because of an issue known as entanglement spread, which arises from the fact that entanglement cannot be conditionally discarded without using communication~\cite{Harrow09}. If we change the entanglement resource from maximally entangled states to \textit{embezzling states}~\cite{vanDam03} however, the problem of entanglement spread can be overcome and the Quantum Reverse Shannon Theorem can be proven.\\

\noindent
A \textit{$\delta$-ebit embezzling state} is a bipartite state $\mu_{AB}$ with the feature that the transformation $\mu_{AB}\mapsto\mu_{AB}\otimes\phi_{A'B'}$, where $\phi_{A'B'}$ denotes an ebit (maximally entangled state of Schmidt-rank $2$), can be accomplished up to an error $\delta$ with local operations. Remarkably, $\delta$-ebit embezzling states exist for all $\delta>0$~\cite{vanDam03}.\\

\noindent
In this paper we present a proof of the Quantum Reverse Shannon Theorem based on \textit{one-shot information theory}. In quantum information theory one usually makes the assumption that the resources are \textit{independent and identically distributed (iid)} and is interested in asymptotic rates. In this case many operational quantities can be expressed in terms of a few information measures (which are usually based on the von Neumann entropy). In contrast to this, \textit{one-shot information theory} applies to arbitrary (structureless) resources. For example, in the context of source coding, it is possible to analyze scenarios where only finitely many, possible correlated messaged are encoded. For this the \textit{smooth entropy formalism} was introduced by Renner et al.~\cite{Ren05, wolf:04n, renner:04n}. Smooth entropy measures have properties similar to the ones of the von Neumann entropy and like in the iid case many operational interpretations are known~\cite{Berta08,Wullschleger08,koenig-2008,datta-2008-2,datta-2008,Ren05,datta-2008-4,datta-2008-5,datta-2008-6,datta-2008-7,wang10,renes10,Renes10_2,tomamichel10,buscemi10_1,buscemi10_2}.\\

\noindent
For our proof of the Quantum Reverse Shannon Theorem we work in this smooth entropy formalism and use a one-shot version for \textit{Quantum State Merging} and its dual \textit{Quantum State Splitting} as well as the \textit{Post-Selection Technique} for quantum channels. As in the original proof of the Quantum Reverse Shannon Theorem~\cite{Bennett06} we need embezzling states.\\

\noindent    
\emph{Quantum State Merging} was introduced by Horodecki et al.~in~\cite{Horodecki05,Horodecki06}. It has since become an important tool in quantum information processing and was subsequently reformulated in~\cite{Abeyesinghe06}, where it is called \textit{mother protocol}. Quantum State Merging corresponds to the quantum generalization of classical Slepian and Wolf coding~\cite{Slepian71}. For its description, one considers a sender system, traditionally called Alice, a receiver system, Bob, as well as a reference system $R$. In Quantum State Merging,  Alice, Bob, and the reference are initially in a joint pure state $\rho_{ABR}$ and one asks how much of a given resource, such as classical or quantum communication or entanglement, is needed in order to move the $A$-part of $\rho_{ABR}$ from Alice to Bob. The dual of this, called \emph{Quantum State Splitting}, addresses the problem of how much of a given resource, such as classical or quantum communication or entanglement, is needed in order to transfer the $A'$-part of a pure state $\rho_{AA'R}$, where part $AA'$ is initially with Alice, from Alice to Bob.\\

\noindent
The \emph{Post-Selection Technique} was introduced in~\cite{ChristKoenRennerPostSelect} and is a tool in order to estimate the closeness of two completely positive and trace preserving (CPTP) maps that act symmetrically on an $n$-partite system, in the metric induced by the diamond norm, the dual of the completely bounded norm~\cite{Kitaev97}. The definition of this norm involves a maximization over all possible inputs to the joint mapping consisting of the CPTP map tensored with an identity map on an outside system. The Post-Selection Technique allows to drop this maximization. In fact, it suffices to consider a single de Finetti type input state, i.e.~a state which consist of $n$ identical and independent copies of an (unknown) state on a single subsystem. The technique was applied in quantum cryptography to show that security of discrete-variable quantum key distribution against a restricted type of attacks, called collective attacks, already implies security against the most general attacks~\cite{ChristKoenRennerPostSelect}.\\

\noindent
Our proof of the Quantum Reverse Shannon Theorem is based on the following idea. Let $\cE_{A\rightarrow B}$ be a quantum channel that takes inputs $\rho_{A}$ on Alice's side and outputs $\cE_{A\rightarrow B}(\rho_{A})$ on Bob's side. To find a way to simulate this quantum channel, it is useful to think of $\cE_{A\rightarrow B}$ as
\begin{align*}
\cE_{A\rightarrow B}(\rho_{A})=\tr_{C}\left[\left(U_{A\rightarrow BC}\right)\rho_{A}\left(U_{A\rightarrow BC}\right)^{\dagger}\right]\ ,
\end{align*}
where $C$ is an additional register and $U_{A\rightarrow BC}$ is some isometry from $A$ to $BC$. This is the Stinespring dilation~\cite{Stinespring55}. Now the idea is to first simulate the isometry $U_{A\rightarrow BC}$ locally at Alice's side, resulting in $\rho_{BC}=(U_{A\rightarrow BC})\rho_{A}(U_{A\rightarrow BC})^{\dagger}$, and in a second step use Quantum State Splitting to do an optimal state transfer of the $B$-part to Bob's side, such that he holds $\rho_{B}=\cE_{A\rightarrow B}(\rho_{A})$ in the end. This simulates the channel $\cE_{A\rightarrow B}$. To prove the Quantum Reverse Shannon Theorem, it is then sufficient to show that the classical communication rate of the Quantum State Splitting protocol is $C_{E}(\cE)$.\\

\noindent
We realize this idea in two steps. Firstly, we propose a new version of Quantum State Splitting (since the known protocols are not good enough to achieve a classical communication rate of $C_{E}(\cE)$), which is based on one-shot Quantum State Merging~\cite{Berta08}. For the analysis we require a \textit{decoupling theorem}, which is optimal in the one-shot case~\cite{Wullschleger08,Dup09}. This means that the decoupling can be achieved optimally even if only a single instance of a quantum state is available. Secondly, we use the Post-Selection Technique to show that our protocol for Quantum State Splitting is sufficient to asymptotically simulate the channel $\cE_{A\rightarrow B}$ with a classical communication rate of $C_{E}(\cE)$. This then completes the proof of the Quantum Reverse Shannon Theorem.\\

\noindent
Our paper is structured as follows. In Section~\ref{entropy} we introduce our notation and give some definitions. In particular, we review the relevant smooth entropy measures. Our results about Quantum State Splitting are then discussed in Section~\ref{splitting}. Finally, we give our proof of the Quantum Reverse Shannon Theorem in Section~\ref{shannonr}. The argument uses various technical statements (e.g.~properties of smooth entropies), which are proved in the appendix.

\section{Smooth Entropy Measures -- Notation and Definitions} \label{entropy}

We assume that all Hilbert spaces, in the following denoted $\cH$, are finite-dimensional. The dimension of $\cH_{A}$ is denoted by $|A|$. The set of linear operators on $\cH$ is denoted by $\cL(\cH)$ and the set of positive semi-definite operators on $\cH$ is denoted by $\cP(\cH)$. We define the sets of subnormalized states $\cS_{\leq}(\cH)=\{\rho\in\cP(\cH):\tr[\rho]\leq1\}$ and normalized states $\cS_{=}(\cH)=\{\rho\in\cP(\cH):\tr[\rho]=1\})$.\\

\noindent
The tensor product of $\cH_{A}$ and $\cH_{B}$ is denoted by $\cH_{AB}\equiv\cH_{A}\otimes\cH_{B}$. Given a multipartite operator $\rho_{AB}\in\cP(\cH_{AB})$, we write $\rho_{A}=\tr_{B}[\rho_{AB}]$ for the corresponding reduced operator. For $M_{A}\in\cL(\cH_{A})$, we write $M_{A}\equiv M_{A}\otimes\1_{B}$ for the enlargement on any $\cH_{AB}$, where $\1_{B}$ denotes the identity in $\cL(\cH_{B})$. Isometries from $\cH_{A}$ to $\cH_{B}$ are denoted by $V_{A\rightarrow B}$.\\

\noindent
For $\cH_{A}$, $\cH_{B}$ with bases $\{\ket{i}_{A}\}_{i=1}^{|A|}$, $\{\ket{i}_{B}\}_{i=1}^{|B|}$ and $|A|=|B|$, the canonical identity mapping from $\cL(\cH_{A})$ to $\cL(\cH_{B})$ with respect to these bases is denoted by $\cI_{A\rightarrow B}$, i.e.~$\cI_{A\rightarrow B}(\ket{i}\bra{j}_{A})=\ket{i}\bra{j}_{B}$. A linear map $\cE_{A\rightarrow B}:\cL(\cH_{A})\rightarrow\cL(\cH_{B})$ is positive if $\cE_{A\rightarrow B}(\rho_{A})\in\cP(\cH_{B})$ for all $\rho_{A}\in\cP(\cH_{A})$. It is completely positive if the map $(\cE_{A\rightarrow B}\otimes\cI_{C\rightarrow C})$ is positive for all $\cH_{C}$. Completely positive and trace preserving maps are called CPTP maps or quantum channels.\\

\noindent
The support of $\rho\in\cP(\cH)$ is denoted by $\mathrm{supp}(\rho)$, the projector onto $\mathrm{supp}(\rho)$ is denoted by $\rho^{0}$ and $\tr\left[\rho^{0}\right]=\mathrm{rank}(\rho)$, the rank of $\rho$. For $\rho\in\cP(\cH)$ we write $\|\rho\|_{\infty}$ for the operator norm of $\rho$, which is equal to the maximum eigenvalue of $\rho$.\\

\noindent
Recall the following standard definitions. The \emph{von Neumann entropy of} $\rho\in\cS_{=}(\cH)$ is defined as\footnote{All logarithms are taken to base 2.}
\begin{align}
H(\rho)=-\tr\left[\rho\log\rho\right]\ .
\end{align}
The \emph{quantum relative entropy of $\rho\in\cS_{\leq}(\cH)$ with respect to $\sigma\in\cP(\cH)$} is given by
\begin{align}
D(\rho\|\sigma)=\tr[\rho\log\rho]-\tr[\rho\log\sigma]
\end{align}
if $\mathrm{supp}(\rho)\subseteq\mathrm{supp}(\sigma)$ and $\infty$ otherwise. The \emph{conditional von Neumann entropy of $A$ given $B$} for $\rho_{AB}\in\cS_{=}(\cH)$ is defined as
\begin{align}
H(A|B)_{\rho}=-D(\rho_{AB}\|\1_{A}\otimes\rho_{B})\ .
\end{align}
The \emph{mutual information between $A$ and $B$} for $\rho_{AB}\in\cS_{=}(\cH)$ is given by
\begin{align}
I(A:B)_{\rho}=D(\rho_{AB}\|\rho_{A}\otimes\rho_{B})\ .
\end{align}
Note that we can also write
\begin{align*}
& H(A|B)_{\rho}=-\inf_{\sigma_{B}\in\cS_{=}(\cH_{B})}D(\rho_{AB}\|\1_{A}\otimes\sigma_{B})\\
& I(A:B)_{\rho}=\inf_{\sigma_{B}\in\cS_{=}(\cH_{B})}D(\rho_{AB}\|\rho_{A}\otimes\sigma_{B})\ .
\end{align*}
We now give the definitions of the smooth entropy measures that we need in this work. In Appendix~\ref{app:entropy} some basic properties are summarized. For a more detailed discussion of the smooth entropy formalism we refer to~\cite{koenig-2008, Ren05, Tomamichel08, Tomamichel09, datta-2008-2}.\\

\noindent
Following Datta~\cite{datta-2008-2} we define the \emph{max-relative entropy of $\rho\in\cS_{\leq}(\cH)$ with respect to $\sigma\in\cP(\cH)$} as
\begin{align}
D_{\max}(\rho\|\sigma)=\inf\{\lambda\in\mathbb{R}:2^{\lambda}\cdot\sigma\geq\rho\}\ .
\end{align}
The \emph{conditional min-entropy of $A$ given $B$} for $\rho_{AB}\in\cS_{\leq}(\cH_{AB})$ is defined as
\begin{align}
H_{\min}(A|B)_{\rho}=-\inf_{\sigma_{B}\in\cS_{=}(\cH_{B})}D_{\max}(\rho_{AB}\|\1_{A}\otimes\sigma_{B})\ .
\end{align}
In the special case where $B$ is trivial, we get $H_{\min}(A)_{\rho}=-\log\|\rho_{A}\|_{\infty}$.\\

\noindent
The \emph{max-information that $B$ has about $A$} for $\rho_{AB}\in\cS_{\leq}(\cH_{AB})$ is defined as
\begin{align}
I_{\max}(A:B)_{\rho}=\inf_{\sigma_{B}\in\cS_{=}(\cH_{B})}D_{\max}(\rho_{AB}\|\rho_{A}\otimes\sigma_{B})\ .
\end{align}
Note that unlike the mutual information, this definition is not symmetric.\\

\noindent
The smooth entropy measures are defined by extremizing the non-smooth measures over a set of nearby states, where our notion of nearby is expressed in terms of the \emph{purified distance}. For $\rho$, $\sigma\in\cS_{\leq}(\cH)$ it is defined as~\cite[Definition 4]{Tomamichel09}
\begin{align}
P(\rho,\sigma)=\sqrt{1-\genFid^{2}(\rho,\sigma)}\ ,
\end{align}
where $\genFid(\cdot\,,\cdot)$ denotes the \emph{generalized fidelity} (which equals the standard fidelity\footnote{The fidelity between $\rho,\sigma\in\cS_{\leq}(\cH)$ is defined by $F(\rho,\sigma)=\left\|\sqrt{\rho}\sqrt{\sigma}\right\|_1$, where $\left\|\Gamma\right\|_{1}=\tr\left[\sqrt{\Gamma\Gamma^{\dagger}}\right]$.} if at least one of the states is normalized),
\begin{align}
\genFid(\rho,\sigma)=\bigl\|\sqrt{\rho\oplus(1-\tr\rho)}\sqrt{\sigma\oplus(1-\tr\sigma)}\bigr\|_1=F(\rho,\sigma)+\sqrt{(1-\tr\rho)(1-\tr\sigma)}\ .
\end{align}
The purified distance is a distance measure on $\cS_{\leq}(\cH)$~\cite[Lemma 5]{Tomamichel09}, in particular, it satisfies the triangle inequality $P(\rho,\sigma)\leq P(\rho,\omega)+P(\omega,\sigma)$ for $\rho,\sigma,\omega\in\cS_{\leq}(\cH)$. $P(\rho,\sigma)$ corresponds to one half times the minimum trace distance\footnote{The trace distance between $\rho,\sigma\in\cS_{\leq}(\cH)$ is defined by $\left\|\rho-\sigma\right\|_1$. The trace distance is often defined with an additional factor one half; we choose not to do this.} between purifications of $\rho$ and $\sigma$.\\

\noindent
Henceforth we call $\rho$, $\sigma\in\cS_{\leq}(\cH)$ $\eps$-close if $P(\rho,\sigma)\leq\eps$ and denote this by $\rho\approx_{\eps}\sigma$. We use the purified distance to specify a ball of subnormalized density operators around $\rho\in\cS_{\leq}(\cH)$:
\begin{align}
\cB^{\eps}(\rho)=\{\bar{\rho}\in\cS_\leq(\cH):P(\rho,\bar{\rho})\leq\eps\}\ .
\end{align}
Miscellaneous properties of the purified distance that we use for our proof are stated in Appendix~\ref{app:purdist}. For a further discussion we refer to~\cite{Tomamichel09}.\\

\noindent
For $\eps\geq0$, the \emph{smooth conditional min-entropy of $A$ given $B$} for $\rho_{AB}\in\cS_{\leq}(\cH_{AB})$ is defined as
\begin{align}
H_{\min}^{\eps}(A|B)_{\rho}=\sup_{\bar{\rho}_{AB}\in\cB^{\eps}(\rho_{AB})}H_{\min}(A|B)_{\bar{\rho}}\ .
\end{align}
The \emph{smooth max-information that $B$ has about $A$} for $\rho_{AB}\in\cS_{\leq}(\cH_{AB})$ is defined as
\begin{align}
I_{\max}^{\eps}(A:B)_{\rho}=\inf_{\bar{\rho}_{AB}\in\cB^{\eps}(\rho_{AB})}I_{\max}(A:B)_{\bar{\rho}}\ .
\label{iiii}
\end{align}
The smooth entropy measure can be seen as a generalization of its corresponding von Neumann quantity in the sense that the latter can be retrieved asymptotically by evaluating the smooth entropy measure on iid states (cf.~Remark~\ref{newremark} and Corollary~\ref{haus2}). In Section~\ref{splitting} we give an operational meaning to the smooth max-information (Theorem~\ref{thi} and Theorem~\ref{convi}).\footnote{For an operational meaning of the smooth conditional min-entropy see e.g.~\cite{Ren05, Berta08}.}\\

\noindent
Since all Hilbert spaces in this paper are assumed to be finite dimensional, we are allowed to replace the infima by minima and the suprema by maxima in all the  definitions of this section. We will do so in the following.

\section{Quantum State Splitting}\label{splitting}

The main goal of this section is to prove that there exists a one-shot Quantum State Splitting protocol (Theorem~\ref{thi}) that is optimal in terms of its quantum communication cost (Theorem~\ref{convi}). The protocol is obtained by inverting a one-shot Quantum State Merging protocol.\\

\noindent
The main technical ingredient for the construction of these protocols is the following decoupling theorem (Theorem~\ref{thm:decoupling}). The proof of the decoupling theorem can be found in Appendix~\ref{app:decoupling}.

\noindent
\begin{theorem}\label{thm:decoupling}
Let $\eps>0$, $\rho_{AR}\in\cS_{\leq}(\cH_{AR})$ and consider a decomposition of the system $A$ into two subsystems $A_{1}$ and $A_{2}$. Furthermore define $\sigma_{A_{1}R}(U)=\tr_{A_{2}}\left[(U\otimes\1_{R})\rho_{AR} (U^\dagger\otimes\1_{R})\right]$. If
\begin{align}
\log |A_{1}|\leq\frac{\log|A|+H_{\min}(A|R)_\rho}{2}-\log\frac{1}{\eps} \ ,
\label{eq:decouplingcondition}
\end{align}
then
\begin{align}
\int_{U(A)}\left\|\sigma_{A_{1}R}(U)-\frac{\1_{A_{1}}}{|A_{1}|}\otimes\rho_{R}\right\|_{1}dU\leq\eps\ ,
\end{align}
where $dU$ is the Haar measure over the unitaries on system $A$, normalized to $\int dU=1$.
\end{theorem}

\noindent
An excellent introduction into the subject of decoupling can be found in~\cite{Hayden11}. Note that our decoupling theorem (Theorem~\ref{thm:decoupling}) can be seen as a special case of a more general decoupling theorem~\cite{Wullschleger08,Dup09}. It is possible to formulate the decoupling criterion in Theorem~\ref{thm:decoupling} more generally in terms of smooth entropies, which is then optimal in the most general one-shot case~\cite{Wullschleger08, Berta08}.\\

\noindent
Quantum State Merging, Quantum State Splitting, and other related quantum information processing primitives are discussed in detail in~\cite{Horodecki05,Horodecki06,Abeyesinghe06,Oppenheim08,Berta08}. Note that we are not only interested in asymptotic rates, but in (tight) one-shot protocols. This is reflected by the following definitions.

\noindent
\begin{definition}[Quantum State Merging]
Consider a bipartite system with parties Alice and Bob. Let $\eps>0$ and $\rho_{ABR}=\ket{\rho}\bra{\rho}_{ABR}\in\cS_{\leq}(\cH_{ABR})$, where Alice controls $A$, Bob $B$ and $R$ is a reference system. A CPTP map $\cE$ is called \emph{$\eps$-error Quantum State Merging of $\rho_{ABR}$} if it consists of applying local operations at Alice's side, local operations at Bob's side, sending $q$ qubits from Alice to Bob and outputs a state
\begin{align}
(\cE\otimes\cI_{R})(\rho_{ABR})\approx_{\eps}\rho_{B'BR}\otimes\proj{\phi_{L}}_{A_{1}B_{1}}\ ,
\end{align}
where $\ket{\phi_{L}}\bra{\phi_{L}}_{A_{1}B_{1}}$ is a maximally entangled state of Schmidt-rank $L$ and $\rho_{B'BR}=(\cI_{A\rightarrow B'}\otimes\cI_{BR})\rho_{ABR}$. $q$ is called quantum communication cost and $e=\lfloor\log L\rfloor$ entanglement gain.
\end{definition}

\noindent
Quantum State Merging is also called \textit{Fully Quantum Slepian Wolf (FQSW)} or \textit{mother protocol}~\cite{Abeyesinghe06}.

\noindent
\begin{lemma}
Let $\eps>0$ and $\rho_{ABR}=\ket{\rho}\bra{\rho}_{ABR}\in\cS_{\leq}(\cH_{ABR})$. Then there exists an $\eps$-error Quantum State Merging protocol for $\rho_{ABR}$ with a quantum communication cost of
\begin{align}
q=\left\lceil\frac{1}{2}\left(H_{0}(A)_{\rho}-H_{\min}(A|R)_{\rho}\right)+2\cdot\log\frac{1}{\eps}\right\rceil
\end{align}
and an entanglement gain of
\begin{align}
e=\left\lfloor\frac{1}{2}\left(H_{0}(A)_{\rho}+H_{\min}(A|R)_{\rho}\right)-2\cdot\log\frac{1}{\eps}\right\rfloor\ ,
\end{align}
where $H_{0}(A)_{\rho}=\log\mathrm{rank}(\rho_{A})$.
\label{state3}
\end{lemma}

\noindent
\begin{proof}
The intuition is as follows (cf.~Figure~\ref{A}). First Alice applies a unitary $U_{A\rightarrow A_{1}A_{2}}$. After this she sends $A_{2}$ to Bob who then performs a local isometry $V_{A_{2}B\rightarrow B'BB_{1}}$. We choose $U_{A\rightarrow A_{1}A_{2}}$ such that it decouples $A_{1}$ from the reference $R$. After sending the $A_{2}$-part to Bob, the state on $A_{1}R$ is given by $\frac{\1_{A_{1}}}{|A_{1}|}\otimes\rho_{R}$ and Bob holds a purification of this. But $\frac{\1_{A_{1}}}{|A_{1}|}\otimes\rho_{R}$ is the reduced state of $\rho_{B'BR}\otimes\proj{\phi_{L}}_{A_{1}B_{1}}$ and since all purifications are equal up to local isometries, there exists an isometry $V_{A_{2}B\rightarrow B'BB_{1}}$ on Bob's side that transforms the state into $\rho_{B'BR}\otimes\proj{\phi_{L}}_{A_{1}B_{1}}$.\\

\noindent
More formally, let $A=A_{1}A_{2}$ with $\log|A_{2}|=\lceil\frac{1}{2}(\log|A|-H_{\min}(A|R)_{\rho})+2\cdot\log\frac{1}{\eps}\rceil$. According to the decoupling theorem (Theorem~\ref{thm:decoupling}), there exists a unitary $U_{A\rightarrow A_{1}A_{2}}$ such that for $\sigma_{A_{1}A_{2}BR}=(U_{A\rightarrow A_{1}A_{2}}\otimes\1_{BR})\rho_{ABR}(U_{A\rightarrow A_{1}A_{2}}^{\dagger}\otimes\1_{BR})$, $\left\|\sigma_{A_{1}R}-\frac{\1_{A_{1}}}{|A_{1}|}\otimes\rho_{R}\right\|_{1}\leq\eps^{2}$. By an upper bound of the purified distance in terms of the trace distance (Lemma~\ref{a:1}) this implies $\sigma_{A_{1}R}\approx_{\eps}\frac{\1_{A_{1}}}{|A_{1}|}\otimes\rho_{R}$.\\

\noindent
We apply this unitary $U_{A\rightarrow A_{1}A_{2}}$ and then send $A_{2}$ to Bob; therefore $q=\left\lceil\frac{1}{2}(\log|A|-H_{\min}(A|R)_{\rho})+2\cdot\log\frac{1}{\eps}\right\rceil$. Uhlmann's theorem~\cite{Uhlmann76, Jozsa94} tells us that there exists an isometry $V_{A_{2}B\rightarrow B'BB_{1}}$ such that
\begin{align*}
P\left(\sigma_{A_{1}R},\frac{\1_{A_{1}}}{|A_{1}|}\otimes\rho_{R}\right)=P\left(\left(\1_{A_{1}R}\otimes V_{A_{2}B\rightarrow B'BB_{1}}\right)\sigma_{A_{1}A_{2}BR}(\1_{A_{1}R}\otimes V_{A_{2}B\rightarrow B'BB_{1}})^{\dagger},\proj{\phi_{L}}_{A_{1}B_{1}}\otimes\rho_{B'BR}\right)\ .
\end{align*}
Hence the entanglement gain is given by $e=\left\lfloor\frac{1}{2}(\log|A|+H_{\min}(A|R)_{\rho})-2\cdot\log\frac{1}{\eps}\right\rfloor$. Now if $\rho_{A}$ has full rank this is already what we want. In general $\log\tr\left[\rho_{A}^{0}\right]=\log|\hat{A}|\leq\log|A|$. But in this case we can restrict the Hilbert space $\cH_{A}$ to the subspace $\cH_{\hat{A}}$, on which $\rho_{A}$ has full rank.
\end{proof}

\noindent
\begin{definition}[Quantum State Splitting with maximally entangled states]
Consider a bipartite system with parties Alice and Bob. Let $\eps>0$ and $\rho_{AA'R}=\ket{\rho}\bra{\rho}_{AA'R}\in\cS_{\leq}(\cH_{AA'R})$, where Alice controls $AA'$ and $R$ is a reference system. Furthermore let $\ket{\phi_{K}}\bra{\phi_{K}}_{A_{1}B_{1}}$ be a maximally entangled state of Schmidt-rank $K$ between Alice and Bob. A CPTP map $\cE$ is called \emph{$\eps$-error Quantum State Splitting of $\rho_{AA'R}$ with maximally entangled states} if it consists of applying local operations at Alice's side, local operations at Bob's side, sending $q$ qubits from Alice to Bob and outputs a state
\begin{align}
(\cE\otimes\cI_{R})(\rho_{AA'R}\otimes\ket{\phi_{K}}\bra{\phi_{K}}_{A_{1}B_{1}})\approx_{\eps}\rho_{ABR}\ ,
\end{align}
where $\rho_{ABR}=(\cI_{A'\rightarrow B}\otimes\cI_{AR})\rho_{AA'R}$. $q$ is called quantum communication cost and $e=\lfloor\log K\rfloor$ entanglement cost.
\end{definition}

\noindent
This is also called the \emph{Fully Quantum Reverse Shannon (FQRS)} protocol~\cite{Abeyesinghe06}, which is a bit misleading, since there is a danger of confusion with the Quantum Reverse Shannon Theorem.\\

\noindent
Quantum State Splitting with maximally entangled states is dual to Quantum State Merging in the sense that every Quantum State Merging protocol already defines a protocol for Quantum State Splitting with maximally entangled states and vice versa.

\begin{lemma}
Let $\eps>0$ and $\rho_{AA'R}=\proj{\rho}_{AA'R}\in\cS_{\leq}(\cH_{AA'R})$. Then there exists an $\eps$-error Quantum State Splitting protocol with maximally entangled states for $\rho_{AA'R}$ with a quantum communication cost of
\begin{align}
q=\left\lceil\frac{1}{2}\left(H_{0}(A')_{\rho}-H_{\min}(A'|R)_{\rho}\right)+2\cdot\log\frac{1}{\eps}\right\rceil
\end{align}
and an entanglement cost of
\begin{align}
e=\left\lfloor\frac{1}{2}(H_{0}(A')_{\rho}+H_{\min}(A'|R)_{\rho})-2\cdot\log\frac{1}{\eps}\right\rfloor\ ,
\end{align}
where $H_{0}(A')_{\rho}=\log\mathrm{rank}(\rho_{A'})$.
\label{hzz}
\end{lemma}

\noindent
\begin{proof}
The Quantum State Splitting protocol with maximally entangled states is defined by running the Quantum State Merging protocol of Theorem~\ref{state3} backwards (see Figure~\ref{A}). The claim then follows from Theorem~\ref{state3}.
\end{proof}

\noindent
\begin{figure}
\includegraphics[width=0.9\linewidth]{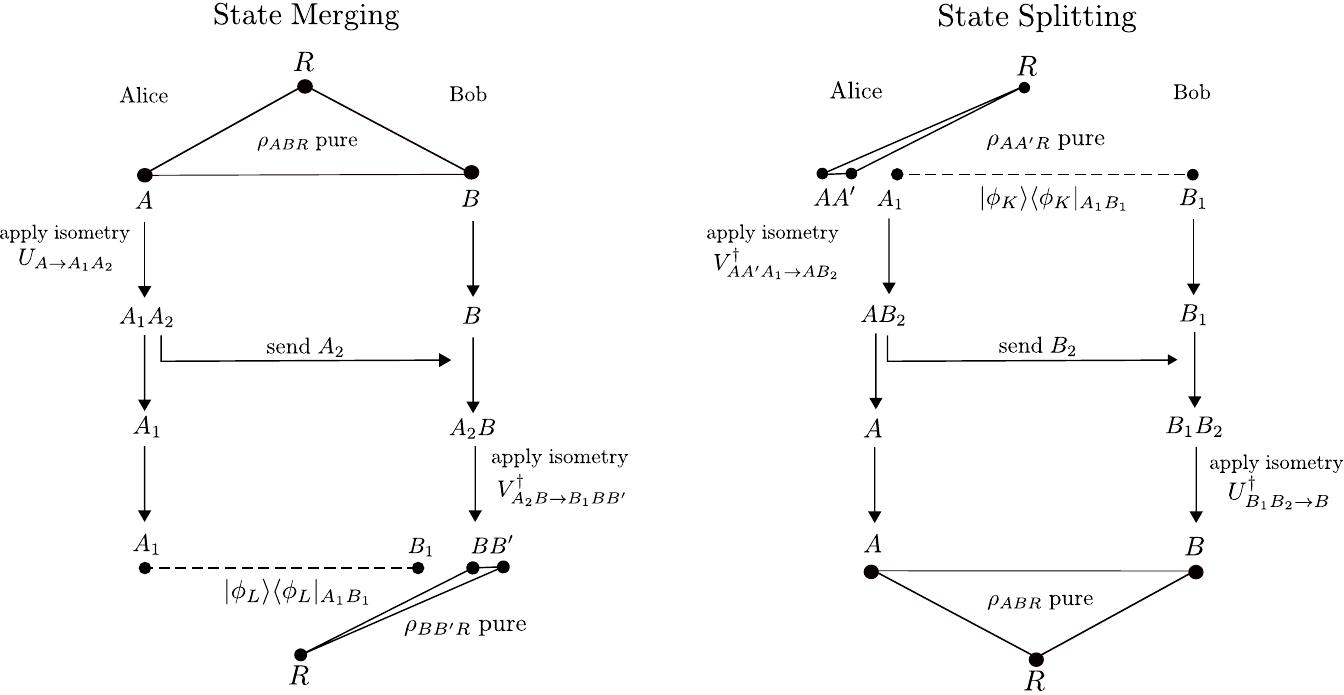}
\caption{From the protocol for Quantum State Merging, which we describe in Lemma~\ref{state3}, we get a protocol for Quantum State Splitting with maximally entangled states. All we have to do is to run the Quantum State Merging protocol backwards.}
\label{A}
\end{figure}

\noindent
In order to obtain a Quantum State Splitting protocol that is optimal in terms of its quantum communication cost, we need to replace the maximally entangled states by embezzling states~\cite{vanDam03}.

\noindent
\begin{definition}
Let $\delta>0$. A state $\mu_{AB}\in\cS_{=}(\cH_{AB})$ is called a \textit{$\delta$-ebit embezzling state} if there exist isometries $X_{A\rightarrow AA'}$ and $X_{B\rightarrow BB'}$ such that
 \begin{align}
(X_{A\rightarrow AA'}\otimes X_{B\rightarrow BB'})\mu_{AB}(X_{A\rightarrow AA'}&\otimes X_{B\rightarrow BB'})^{\dagger}\approx_{\delta}\mu_{AB}\otimes\proj{\phi}_{A'B'}\ ,
\end{align}
where $\proj{\phi}_{A'B'}\in\cS_{=}(\cH_{A'B'})$ denotes an ebit (maximally entangled state of Schmidt rank $2$).
\label{p}
\end{definition}

\begin{proposition}\cite{vanDam03}
$\delta$-ebit embezzling states exist for all $\delta>0$.
\end{proposition}

\noindent
We would like to highlight two interesting examples. For the first example consider the state $\proj{\mu_{m}}_{A^{m}B^{m}}\in\cS_{=}(\cH_{A^{m}B^{m}})$ defined by
\begin{align*}
\ket{\mu_{m}}_{A^{m}B^{m}}=C\cdot\sum_{j=0}^{m-1}\ket{\varphi}_{AB}^{\otimes j}\otimes\ket{\phi}_{AB}^{\otimes(m-j)}\ ,
\end{align*}
where $\proj{\varphi}_{AB}\in\cS_{=}(\cH_{AB})$, $\proj{\phi}_{AB}\in\cS_{=}(\cH_{AB})$ denotes an ebit and $C$ is such that $\proj{\mu_{m}}_{A^{m}B^{m}}$ is normalized. Note that applying the cyclic shift operator $U_{A_{0}A^{m}}$ that sends $A_{i}\rightarrow A_{i+1}$ at Alice's side (modulo $m+1$) and the corresponding cyclic shift operator $U_{B_{0}B^{m}}$ at Bob's side maps $\ket{\varphi}_{A_{0}B_{0}}\otimes\ket{\mu_{m}}_{A^{m}B^{m}}$ to $\ket{\phi}_{A_{0}B_{0}}\otimes\ket{\mu_{m}}_{A^{m}B^{m}}$ up to an accuracy of $\sqrt{\frac{2}{m}}$~\cite{Leung_08}. For the choice $\ket{\varphi}_{AB}=\ket{\varphi}_{A}\otimes\ket{\varphi}_{B}$, $\proj{\mu_{m}}_{A^{m}B^{m}}$ is a  $\sqrt{\frac{2}{m}}$-ebit embezzling state with the isometries $X_{A^{m}\rightarrow A_{0}A^{m}}=U_{A_{0}A^{m}}\ket{\varphi}_{A_{0}}$ and $X_{B^{m}\rightarrow B_{0}B^{m}}=U_{B_{0}B^{m}}\ket{\varphi}_{B_{0}}$.\\

\noindent
The second example is the state $\proj{\tilde{\mu}_{m}}_{AB}\in\cS_{=}(\cH_{AB})$ defined by
\begin{align*}
\ket{\tilde{\mu}^{m}}_{AB}=(\sum_{j=1}^{2^{m}}\frac{1}{j})^{-1/2}\cdot\sum_{j=1}^{2^{m}}\frac{1}{\sqrt{j}}\ket{j}_{A}\otimes\ket{j}_{B}\ .
\end{align*}
It is a $\sqrt{\frac{2}{m}}$-ebit embezzling state~\cite{vanDam03}.\footnote{The state $\proj{\tilde{\mu}_{m}}_{AB}$ is even a $(\sqrt{\frac{2}{m}},r)$-universal embezzling state~\cite{vanDam03}. That is, for any $\proj{\varsigma}_{A'B'}\in\cS_{=}(\cH_{A'B'})$ of Schmidt-rank at most $r$, there exist isometries $X_{A\rightarrow AA'}$ and $X_{B\rightarrow BB'}$ such that
 \begin{align*}
(X_{A\rightarrow AA'}\otimes X_{B\rightarrow BB'})\proj{\tilde{\mu}_{m}}_{AB}(X_{A\rightarrow AA'}&\otimes X_{B\rightarrow BB'})^{\dagger}\approx_{\sqrt{\frac{2}{m}}}\proj{\tilde{\mu}_{m}}_{AB}\otimes\proj{\varsigma}_{A'B'}\ .
\end{align*}}

\begin{remark}
By using $\delta$-ebit embezzling states multiple times, it is possible to create maximally entangled states of higher dimension. More precisely, for every $\delta$-ebit embezzling state $\mu_{AB}\in\cS_{=}(\cH_{AB})$ there exist isometries $X_{A\rightarrow AA'}$ and $X_{B\rightarrow BB'}$ such that
 \begin{align}
(X_{A\rightarrow AA'}\otimes X_{B\rightarrow BB'})\mu_{AB}(X_{A\rightarrow AA'}&\otimes X_{B\rightarrow BB'})^{\dagger}\approx_{\delta\cdot\log L}\mu_{AB}\otimes\proj{\phi_{L}}_{A'B'}\ ,
\end{align}
where $\proj{\phi_{L}}_{A'B'}\in\cS_{=}(\cH_{A'B'})$ denotes a maximally entangled state of Schmidt-rank $L$ (with $L$ being a power of $2$).
\label{rmk:emb}
\end{remark}

\noindent
We are now ready to define Quantum State Splitting with embezzling states.

\noindent
\begin{definition}[Quantum State Splitting with embezzling states]
Consider a bipartite system with parties Alice and Bob. Let $\eps>0$, $\delta>0$ and $\rho_{AA'R}=\ket{\rho}\bra{\rho}_{AA'R}\in\cS_{\leq}(\cH_{AA'R})$, where Alice controls $AA'$ and $R$ is a reference system. A CPTP map $\cE$ is called \emph{$\eps$-error Quantum State Splitting of $\rho_{AA'R}$ with a $\delta$-ebit embezzling state} if it consists of applying local operations at Alice's side, local operations at Bob's side, sending $q$ qubits from Alice to Bob, using a $\delta$-ebit embezzling state $\mu_{A_{\mathrm{emb}}B_{\mathrm{emb}}}$, and outputs a state
\begin{align}
(\cE\otimes\cI_{R})(\rho_{AA'R}\otimes\mu_{A_{\mathrm{emb}}B_{\mathrm{emb}}})\approx_{\eps}\rho_{ABR}\ ,
\end{align}
where $\rho_{ABR}=(\cI_{A'\rightarrow B}\otimes\cI_{AR})\rho_{AA'R}$. $q$ is called quantum communication cost.
\label{politesse}
\end{definition}

\noindent
The following theorem about the achievability of Quantum State Splitting with embezzling states (Theorem~\ref{thi}) is the main result of this section. In Section~\ref{shannonr} we use this theorem to prove the Quantum Reverse Shannon Theorem.

\noindent
\begin{theorem}
Let $\eps>0$, $\eps'\geq0$, $\delta>0$ and $\rho_{AA'R}=\proj{\rho}_{AA'R}\in\cS_{\leq}(\cH_{AA'R})$. Then there exists an $(\eps+\eps'+\delta\cdot\log|A'|+|A'|^{-1/2})$-error\footnote{The error term $|A'|^{-1/2}$ can be made arbitrarily small by enlarging the Hilbert space $\cH_{A'}$. Of course this increases the error term $\delta\cdot\log|A'|$, but this can again be compensated with decreasing $\delta$. Enlarging the Hilbert space $\cH_{A'}$ also increases the quantum communication cost~\eqref{eq:thi}, but only slightly.} Quantum State Splitting protocol for $\rho_{AA'R}$ with a $\delta$-ebit embezzling state for a quantum communication cost of
\begin{align}\label{eq:thi}
q\leq\frac{1}{2}I^{\eps'}_{\max}(A':R)_{\rho}+2\cdot\log\frac{1}{\eps}+4+\log\log|A'|\ .
\end{align}
\label{thi}
\end{theorem}

\noindent
\begin{proof}
The idea for the protocol is as follows (cf.~Figure~\ref{C}). First, we disregard the eigenvalues of $\rho_{A'}$ that are smaller then $|A|^{-2}$. This introduces an error $\alpha=|A|^{-1/2}$, but because of the monotonicity of the purified distance (Lemma~\ref{hzhz}), the error at the end of the protocol is still upper bounded by the same $\alpha$. As a next step we let Alice perform a coherent measurement $W$ with roughly $2\cdot\log|A|$ measurement outcomes in the eigenbasis of $\rho_{A'}$. That is, the state after the measurement is of the form $\omega_{AA'RI_{A}}=\proj{\omega}_{AA'RI_{A}}$ with
$$\ket{\omega}_{AA'RI_{A}}=\sum_{i\in I}\sqrt{p_{i}}\ket{\rho^{i}}_{AA'R}\otimes\ket{i}_{I_{A}}\ .$$
Here the index $i$ indicates which measurement outcome occured, $p_{i}$ denotes its probability and $\rho^{i}_{AA'R}=\proj{\rho^{i}}_{AA'R}$ the corresponding post-measurement state.\\

\noindent
Then, conditioned on the index $i$, we use the Quantum State Splitting protocol with maximally entangled states from Lemma~\ref{hzz} for each state $\rho^{i}_{AA'R}$ and denote the corresponding quantum communication cost and entanglement cost by $q_{i}$ and $e_{i}$ respectively. The total amount of quantum communication we need for this is given by $\max_{i}q_{i}$ plus the amount needed to send the register $I_{A}$ (which is of order $\log\log|A|$). In addition, since the different branches of the protocol use different amounts of entanglement, we need to provide a superposition of different (namely $e_{i}$ sized) maximally entangled states. We do this by using embezzling states.\footnote{Note that it is not possible to get such a superposition starting from any amount of maximally entangled states only using local operations. This problem is known as \textit{entanglement spread} and is discussed in~\cite{Harrow09}.}\\

\noindent
As the last step, we undo the initial coherent measurement $W$. This completes the Quantum State Splitting protocol with embezzling states for $\rho_{AA'R}$. All that remains to do, is to bring the expression for the quantum communication cost in the right form. In the following, we describe the proof in detail.\\

\noindent
Let $Q=\lceil2\cdot\log|A'|-1\rceil$, $I=\{0,1,\ldots,Q,(Q+1)\}$ and let $\{P_{A'}^{i}\}_{i\in I}$ be a collection of projectors on $\cH_{A'}$ defined as follows.
$P_{A'}^{Q+1}$ projects on the eigenvalues of $\rho_{A'}$ in $[2^{-2\log|A'|},0]$, $P_{A'}^{Q}$ projects on the eigenvalues of $\rho_{A'}$ in $[2^{-Q},2^{-2\log|A'|}]$ and for $i=0,1,\dots,(Q-1)$, $P_{A'}^{i}$ projects on the eigenvalues of $\rho_{A'}$ in $[2^{-i},2^{-(i+1)}]$.\\

\noindent
Furthermore let $p_{i}=\tr\left[P_{A'}^{i}\rho_{A'}\right]$, $\rho_{AA'R}^{i}=\proj{\rho^{i}}_{AA'R}$ with $\ket{\rho^{i}}_{AA'R}=p_{i}^{-1/2}\cdot P_{A'}^{i}\ket{\rho}_{AA'R}$ and define the state $\bar{\rho}_{AA'R}=\proj{\bar{\rho}}_{AA'R}$ with
$$\ket{\bar{\rho}}_{AA'R}=\Upsilon^{-1/2}\cdot\sum_{i=0}^{Q}\sqrt{p_{i}}\ket{\rho^{i}}_{AA'R}\ ,$$
where $\Upsilon=\sum_{i=0}^{Q}p_{i}$.\\

\noindent
We have
\begin{align}\label{bar}
\bar{\rho}_{AA'R}\approx_{|A'|^{-1/2}}\rho_{AA'R}
\end{align}
as can be seen as follows. We have
\begin{align*}
P(\bar{\rho}_{AA'R},\rho_{AA'R})=\sqrt{1-F^{2}(\bar{\rho}_{AA'R},\rho_{AA'R})}=\sqrt{1-|\braket{\bar{\rho}}{\rho}_{AA'R}|^{2}}=\sqrt{1-\sum_{i=0}^{Q}p_{i}}=\sqrt{p_{Q+1}}\ .
\end{align*}
But because at most $|A'|$ eigenvalues of $\rho_{A'}$ can lie in $[2^{-2\log|A'|},0]$, each one smaller or equal to $2^{-2\log|A'|}$, we obtain $p_{Q+1}\leq|A'|\cdot2^{-2\log|A'|}=|A'|^{-1}$ and hence $P(\bar{\rho}_{AA'R},\rho_{AA'R})\leq|A'|^{-1/2}$.\\

\noindent
\begin{figure}
\includegraphics[width=0.7\linewidth]{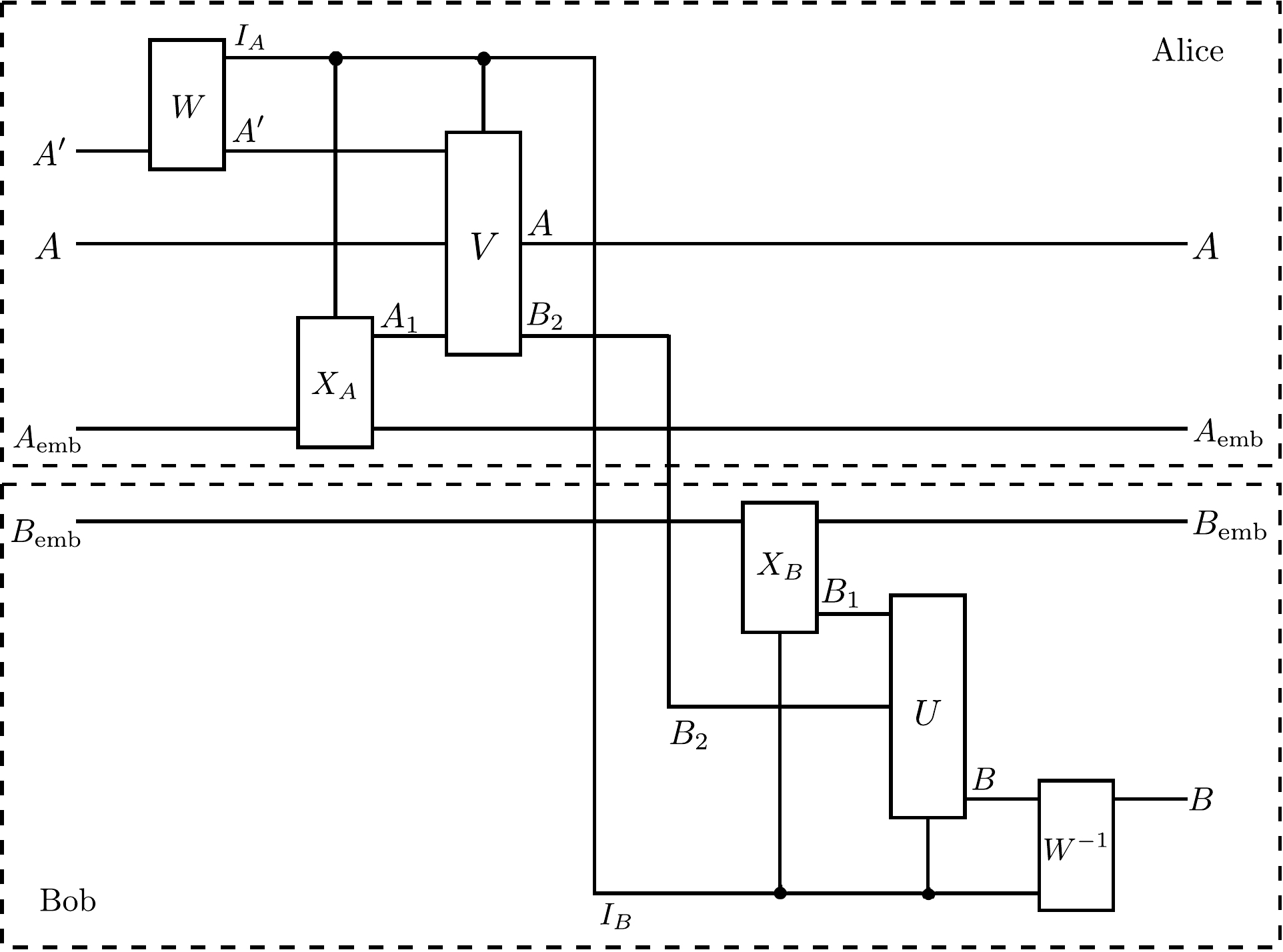}
\caption{A schematic description of our protocol for Quantum State Splitting with embezzling states in the language of the \textit{quantum circuit model}~\cite{Deutsch89,NieChu00Book}. See the text for definitions and a precise description.}
\label{C}
\end{figure}

\noindent
We proceed by defining the operations that we need for the Quantum State Splitting protocol with embezzling states for $\bar{\rho}_{AA'R}$ (cf.~Figure~\ref{C}). Define the isometry
\begin{align}
W_{A'\rightarrow A'I_{A}}=\sum_{i\in I}P_{A'}^{i}\otimes\ket{i}_{I_{A}}\ ,
\label{W}
\end{align}
where the vectors $\ket{i}_{A}$ are mutually orthogonal and $I_{A}$ is at Alice's side. We want to use the $\eps$-error Quantum State Splitting protocol with maximally entangled states from Lemma~\ref{hzz} for each $\rho^{i}_{AA'R}$. For each $i=0,1,\ldots,Q$ this protocol has a quantum communication cost of
$$q_{i}=\left\lceil\frac{1}{2}(H_{0}(A')_{\rho^{i}}-H_{\min}(A'|R)_{\rho^{i}})+2\cdot\log\frac{1}{\eps}\right\rceil$$
and an entanglement cost of
\begin{align}
e_{i}=\left\lfloor\frac{1}{2}(H_{0}(A')_{\rho^{i}}+H_{\min}(A'|R)_{\rho^{i}})-2\cdot\log\frac{1}{\eps}\right\rfloor\ .
\label{ei}
\end{align}
For $A_{1}$ on Alice's side, $B_{1}$ on Bob's side and $A_{1}^{i}$, $B_{1}^{i}$ $2^{e_{i}}$-dimensional subspaces of $A_{1}$, $B_{1}$ respectively, the Quantum State Splitting protocol from Lemma~\ref{hzz} has the following form: apply the isometry $V^{i}_{AA'A_{1}^{i}\rightarrow AB_{2}^{i}}$ on Alice's side, send $B_{2}^{i}$ from Alice to Bob (for a quantum communication cost of $q_{i}$) and then apply the isometry $U^{i}_{B_{1}^{i}B_{2}^{i}\rightarrow B}$ on Bob's side.\\

\noindent
As a next ingredient to the protocol, we define the isometries that supply the maximally entangled states of size $e_{i}$. For $i=0,1,\ldots,Q$, let $X_{A_{\mathrm{emb}}\rightarrow A_{\mathrm{emb}}A_{1}^{i}}^{i}$ and $X_{B_{\mathrm{emb}}\rightarrow B_{\mathrm{emb}}B_{1}^{i}}^{i}$ be the isometries at Alice's and Bob's side respectively, that embezzle, with accuracy $\delta\cdot e_{i}$, a maximally entangled state of dimension $e_{i}$ out of the embezzling state and put it in $A_{1}^{i}B_{1}^{i}$.\\

\noindent
We are now ready to put the isometries together and give the protocol for Quantum State Splitting with embezzling states for $\bar{\rho}_{AA'R}$ (cf.~Figure~\ref{C}). Alice applies the isometry $W_{A'\rightarrow A'I_{A}}$ followed by the isometry
$$X_{A_{\mathrm{emb}I_{A}}\rightarrow A_{\mathrm{emb}}A_{1}I_{A}}=\sum_{i=1}^{Q}X_{A_{\mathrm{emb}}\rightarrow A_{\mathrm{emb}}A_{1}^{i}}^{i}\otimes\proj{i}_{I_{A}}$$
and the isometry 
\begin{align*}
V_{AA'A_{1}I_{A}\rightarrow AB_{2}I_{A}}=\sum_{i=0}^{Q}V_{AA'A_{1}^{i}\rightarrow AB_{2}^{i}}^{i}\otimes\proj{i}_{I_{A}}\ .
\end{align*}
Afterwards she sends $I_{A}$ and $B_{2}$, that is
$$q=\max_{i}\left\lceil\frac{1}{2}(H_{0}(A')_{\rho^{i}}-H_{\min}(A'|R)_{\rho^{i}})+2\cdot\log\frac{1}{\eps}\right\rceil+\log\left\lceil2\cdot\log|A'|\right\rceil$$
qubits to Bob (where we rename $I_{A}$ to $I_{B}$). Then Bob applies the isometry
$$X_{B_{\mathrm{emb}I_{B}}\rightarrow B_{\mathrm{emb}}B_{1}I_{B}}=\sum_{i=1}^{Q}X_{B_{\mathrm{emb}}\rightarrow B_{\mathrm{emb}}B_{1}^{i}}^{i}\otimes\proj{i}_{I_{B}}$$
followed by the isometry
\begin{align}
U_{B_{1}B_{2}I_{B}\rightarrow BI_{B}}=\sum_{i=0}^{Q}U^{i}_{B_{1}^{i}B_{2}^{i}\rightarrow B}\otimes\proj{i}_{I_{B}}\ .
\label{iso99}
\end{align}

\noindent
Next we analyze how the resulting state looks like. By the definition of embezzling states (cf. Definition~\ref{p} and Remark~\ref{rmk:emb}), the monotonicity of the purified distance (Lemma~\ref{hzhz}) and the triangle inequality for the purified distance, we obtain a state $\sigma_{ABRI_{B}}=\ket{\sigma}\bra{\sigma}_{ABRI_{B}}$ with
\begin{align*}
\ket{\sigma}_{ABRI_{B}}=\Upsilon^{-1/2}\cdot\sum_{i=0}^{Q}\sqrt{p_{i}}\ket{\tilde{\rho}^{i}}_{ABR}\otimes\ket{i}_{I_{B}}\ ,
\end{align*}
where $\proj{\tilde{\rho}^{i}}_{ABR}=\tilde{\rho}^{i}_{ABR}\approx_{\eps+\delta\cdot e_{i}}\rho^{i}_{ABR}$ and $\rho^{i}_{ABR}=(\cI_{A'\rightarrow B}\otimes\cI_{AR})\rho^{i}_{AA'R}$ for $i=0,1,\ldots,Q$. The state $\sigma_{ABRI_{B}}$ is close to the state $\omega_{ABRI_{B}}=\ket{\omega}\bra{\omega}_{ABRI_{B}}$ with
\begin{align*}
\ket{\omega}_{ABRI_{B}}=\Upsilon^{-1/2}\cdot\sum_{i=0}^{Q}\sqrt{p_{i}}\ket{\rho^{i}}_{ABR}\otimes\ket{i}_{I_{B}}\ ,
\end{align*}
as can be seen as follows. Because we can assume without lost of generality that all $\spr{\tilde{\rho}^{i}}{\rho^{i}}$ are real and nonnegative,\footnote{This can be done be multiplying the isometries $U_{B_{1}^{i}B_{2}^{i}\rightarrow B}^{i}$ in~\eqref{iso99} with appropriately chosen phase factors.} we obtain
\begin{align*}
P(\sigma_{ABRI_{B}},\omega_{ABRI_{B}})&=\sqrt{1-F^{2}(\sigma_{ABRI_{B}},\omega_{ABRI_{B}})}=\sqrt{1-\left|\spr{\sigma}{\omega}_{ABRI_{B}}\right|^{2}}=\sqrt{1-\left|\frac{1}{\Upsilon}\cdot\sum_{i=0}^{Q}p_{i}\spr{\tilde{\rho}^{i}}{\rho^{i}}_{ABR}\right|^{2}}\\
&=\sqrt{1-\left(\frac{1}{\Upsilon}\cdot\sum_{i=0}^{Q}p_{i}\spr{\tilde{\rho}^{i}}{\rho^{i}}_{ABR}\right)^{2}}=\sqrt{1-\left(\frac{1}{\Upsilon}\cdot\sum_{i=0}^{Q}p_{i}F\left(\tilde{\rho}^{i}_{ABR},\rho^{i}_{ABR}\right)\right)^{2}}\\
&=\sqrt{1-\left(\frac{1}{\Upsilon}\cdot\sum_{i=0}^{Q}p_{i}\sqrt{1-P^{2}\left(\tilde{\rho}^{i}_{ABR},\rho^{i}_{ABR}\right)}\right)^{2}}\\
&\leq\sqrt{1-\left(\frac{1}{\Upsilon}\cdot\sum_{i=0}^{Q}p_{i}\sqrt{1-(\eps+\delta\cdot e_{i})^{2}}\right)^{2}}\leq\sqrt{1-\left(\frac{1}{\Upsilon}\cdot\sum_{i=0}^{Q}p_{i}\sqrt{1-(\eps+\delta\cdot\max_{i}e_{i})^{2}}\right)^{2}}\\
&=\eps+\delta\cdot\max_{i}e_{i}\leq\eps+\delta\cdot\log|A'|\ ,
\end{align*}
where the last inequality follows from~\eqref{ei}. To decode the state $\sigma_{ABRI_{B}}$ to a state that is $(\eps+\delta\cdot\log|A'|)$-close to $\bar{\rho}_{ABR}$, we define the isometry $W_{B\rightarrow BI_{B}}$ analogously to $W_{A'\rightarrow A'I_{A}}$ in~\eqref{W}. Because all isometries are injective, we can define an inverse of $W$ on the image of $W$ (which we denote by $\mathrm{Im}(W)$). The inverse is again an isometry and we denote it by $W^{-1}_{\mathrm{Im}(W)\rightarrow B}$.\\

\noindent
The last step of the protocol is then to apply the CPTP map to the state $\sigma_{ABRI_{B}}$, that first does a measurement on $BI_{B}$ to decide whether $\sigma_{BI_{B}}\in\mathrm{Im}(W)$ or not and then, if $\sigma_{BI_{B}}\in\mathrm{Im}(W)$, applies the isometry $W^{-1}_{\mathrm{Im}(W)\rightarrow B}$ and otherwise maps the state to $\proj{0}_{B}$.\\

\noindent
By the monotonicity of the purified distance (Lemma~\ref{hzhz}) we finally get a state that is $(\eps+\delta\cdot\log|A'|)$-close to $\bar{\rho}_{ABR}$. Hence we showed the existence of an $(\eps+\delta\cdot\log|A'|)$-error Quantum State Splitting protocol with embezzling states for $\bar{\rho}_{AA'R}$ with a quantum communication cost of
\begin{align}\label{cost}
q=\max_{i}\left\lceil\frac{1}{2}\left(H_{0}(A')_{\rho^{i}}-H_{\min}(A'|R)_{\rho^{i}}\right)+2\cdot\log\frac{1}{\eps}\right\rceil+\log\left\lceil2\cdot\log|A'|\right\rceil\ ,
\end{align}
where $i\in\{0,1,\ldots,Q\}$. But by the monotonicity of the purified distance (Lemma~\ref{hzhz}), \eqref{bar} and the triangle inequality for the purifed distance, this implies the existence of an $\left(\eps+\delta\cdot\log|A'|+|A|^{-1/2}\right)$-error Quantum State Splitting protocol with embezzling states for $\rho_{AA'R}$ with a quantum communication cost as in~\eqref{cost}.\\

\noindent
We now proceed with simplifying the expression for the quantum communication cost~\eqref{cost}. We have $H_{0}(A')_{\rho^{i}}\leq H_{\min}(A')_{\rho^{i}}+1$ for $i=0,1,\ldots,Q$ as can be seen as follows. We have
$$2^{-(i+1)}\leq\lambda_{\min}(\rho_{A'}^{i})\leq\frac{1}{\mathrm{rank}\left(\rho_{A'}^{i}\right)}\leq\left\|\rho_{A'}^{i}\right\|_{\infty}\leq2^{-i}\ ,$$
where $\lambda_{\min}(\rho_{A'}^{i})$ denotes the smallest non-zero eigenvalue of $\rho_{A'}^{i}$. Thus $\mathrm{rank}\left(\rho_{A'}^{i}\right)\leq2^{i+1}=2^{i}\cdot2\leq\frac{2}{\left\|\rho_{A'}^{i}\right\|_{\infty}}$ and this is equivalent to the claim.\\

\noindent
Hence we get an $(\eps+\delta\cdot\log|A'|+|A'|^{-1/2})$-error Quantum State Splitting protocol with embezzling states for $\rho_{AA'R}$ with a quantum communication cost of
\begin{align*}
q&=\max_{i}\left\lceil\frac{1}{2}\left(H_{\min}(A')_{\rho^{i}}-H_{\min}(A'|R)_{\rho^{i}}+1\right)+2\cdot\log\frac{1}{\eps}\right\rceil+\log\left\lceil2\cdot\log|A'|\right\rceil\ .
\end{align*}
Using a lower bound for the max-information in terms of min-entropies (Lemma~\ref{tschau2}) and the behavior of the max-information under projective measurements (Lemma~\ref{hehehehe}) we can simplify this to
\begin{align*}
q&\leq\left\lceil\max_{i}\frac{1}{2}I_{\max}(A':R)_{\rho^{i}}+2\cdot\log\frac{1}{\eps}+\frac{1}{2}\right\rceil+\log\left\lceil2\cdot\log|A'|\right\rceil\\
&\leq\left\lceil\frac{1}{2}I_{\max}(A':R)_{\rho}+2\cdot\log\frac{1}{\eps}+\frac{1}{2}\right\rceil+\log\left\lceil2\cdot\log|A'|\right\rceil\ .
\end{align*}
It is then easily seen that
\begin{align*}
q\leq\frac{1}{2}I_{\max}(A':R)_{\rho}+2\cdot\log\frac{1}{\eps}+4+\log\log|A'|\ .
\end{align*}
As the last step, we transform the max-information term in the formula for the quantum communication cost into a smooth max-information. Namely, we can reduce the quantum communication cost if we do not apply the protocol as described above to the state $\rho_{AA'R}$, but pretend that we have another (possibly subnormalized) state $\hat{\rho}_{AA'R}$ that is $\eps'$-close to $\rho_{AA'R}$ and then apply the protocol for $\hat{\rho}_{AA'R}$. By the monotonicity of the purified distance (Lemma~\ref{hzhz}), the additional error term that we get from this is upper bounded by $\eps'$ and by the triangle inequality for the purified distance this results in an accuracy of $\eps+\eps'+\delta\cdot\log|A'|+|A'|^{-1/2}$. But if we minimize $q$ over all $\hat{\rho}_{AA'R}$ that are $\eps'$-close to $\rho_{AA'R}$, we can reduce the quantum communication cost to
\begin{align}\label{final:cost}
q\leq\frac{1}{2}I_{\max}^{\eps'}(A':R)_{\rho}+2\cdot\log\frac{1}{\eps}+4+\log\log|A'|\ .
\end{align}
This shows the existence of an $(\eps+\eps'+\delta\cdot\log|A'|+|A'|^{-1/2})$-error Quantum State Splitting protocol with embezzling states for $\rho_{AA'R}$ for a quantum communication cost as in~\eqref{final:cost}.
\end{proof}

\noindent
The following theorem shows that the quantum communication cost in Theorem~\ref{thi} is optimal up to small additive terms.

\noindent
\begin{theorem}
Let $\eps>0$ and $\rho_{AA'R}=\proj{\rho}_{AA'R}\in\cS_{\leq}(\cH_{AA'R})$. Then the quantum communication cost for any $\eps$-error Quantum State Splitting protocol\footnote{We suppress the mentioning of any entanglement resource, since the statement holds independently of it.} for $\rho_{AA'R}$ is lower bounded by
\begin{align}
q\geq\frac{1}{2}I^{\eps}_{\max}(A':R)_{\rho}\ .
\end{align}
\label{convi}
\end{theorem}

\noindent
\begin{proof}
We have a look at the correlations between Bob and the reference by analyzing the max-information that the reference has about Bob. At the beginning of any protocol, there is no register at Bob's side and therefore the max-information that the reference has about Bob is zero. Since back communication is not allowed, we can assume that the protocol for Quantum State Splitting has the following form: applying local operations at Alice's side, sending qubits from Alice to Bob and then applying local operations at Bob's side. Local operations at Alice's side have no influence on the max-information that the reference has about Bob. By sending $q$ qubits from Alice to Bob, the max-information that the reference has about Bob can increase, but at most by $2q$ (Lemma~\ref{upperb}). By applying local operations at Bob's side the max-information that the reference has about Bob can only decrease (Lemma~\ref{epsilon}). So the max-information that the reference has about Bob is upper bounded by $2q$. Therefore, any state $\omega_{BR}$ at the end of a Quantum State Splitting protocol must satisfy $I_{\max}(B:R)_{\omega}\leq2q$. But we also need $\omega_{BR}\approx_{\eps}\rho_{BR}\equiv(\cI_{A'\rightarrow B}\otimes\cI_{R})(\rho_{A'R})$ by the definition of $\eps$-error Quantum State Splitting (Definition~\ref{politesse}). Using the definition of the smooth max-information, we get
\begin{align*}
q\geq \frac{1}{2}I_{\max}^{\eps}(A':R)_{\rho}\ .
\end{align*}
\end{proof}

\section{The Quantum Reverse Shannon Theorem} \label{shannonr}

This section contains the main result, a proof of the Quantum Reverse Shannon Theorem. The intuition is as follows. Let $\cE_{A\rightarrow B}$ be a quantum channel with
\begin{align*}
\cE_{A\rightarrow B}:\quad&\cS_{=}(\cH_{A})\rightarrow\cS_{=}(\cH_{B})\\
& \rho_{A}\mapsto\cE_{A\rightarrow B}(\rho_{A})\ ,
\end{align*}
where we want to think of subsystem $A$ being at Alice's side and subsystem $B$ being at Bob's side. The Quantum Reverse Shannon Theorem states that if Alice and Bob share embezzling states, they can asymptotically simulate $\cE_{A\rightarrow B}$ only using local operations at Alice's side, local operations at Bob's side, and a classical communication rate (from Alice to Bob) of
\begin{align*}
C_{E}=\max_{\Phi}I(B:R)_{(\cE\otimes\cI)(\Phi)}\ ,
\end{align*}
where $\Phi_{AR}$ is a purification of $\rho_{A}$ and we note that $I(B:R)_{(\cE\otimes\cI)(\Phi)}=H(R)_{\rho}+H(B)_{\cE(\rho)}-H(BR)_{(\cE\otimes\cI)(\Phi)}$.\\

\noindent
Using Stinespring's dilation~\cite{Stinespring55}, we can think of $\cE_{A\rightarrow B}$ as
\begin{align}
\cE_{A\rightarrow B}(\rho_{A})=\tr_{C}\left[(U_{A\rightarrow BC})\rho_{A}(U_{A\rightarrow BC})^{\dagger}\right]\ ,
\label{intuition}
\end{align}
where $C$ is an additional register with $|C|\leq|A||B|$ and $U_{A\rightarrow BC}$ some isometry. The idea of our proof is to first simulate the quantum channel locally at Alice's side, resulting in $\rho_{BC}=(U_{A\rightarrow BC})\rho_{A}(U_{A\rightarrow BC})^{\dagger}$, and then use Quantum State Splitting with embezzling states (Theorem~\ref{thi}) to do an optimal state transfer of the $B$-part to Bob's side, such that he holds $\rho_{B}=\cE_{A\rightarrow B}(\rho_{A})$ in the end. Note that we can replace the quantum communication in the Quantum State Splitting protocol by twice as much classical communication, since we have free entanglement and can therefore use \textit{quantum teleportation}~\cite{teleport}. Although the free entanglement is given in the form of embezzling states, maximally entangled states can be created without any (additional) communication (Definition~\ref{p}).\\

\noindent
More formally, we make the following definitions:
\begin{definition}
Consider a bipartite system with parties Alice and Bob. Let $\eps\geq0$ and $\cE:\cL(\cH_{A})\rightarrow\cL(\cH_{B})$ be a CPTP map, where Alice controls $\cH_{A}$ and Bob $\cH_{B}$. A CPTP map $\cP$ is a \textit{one-shot reverse Shannon simulation for $\cE$ with error $\eps$} if it consists of applying local operations at Alice's side, local operations at Bob's side, sending $c$ classical bits from Alice to Bob, using a $\delta$-ebit embezzling state for some $\delta>0$, and
\begin{align}
\|\cP-\cE\|_{\Diamond}\leq\eps\ ,
\end{align}
where $\|.\|_{\Diamond}$ denotes the diamond norm (Definition~\ref{kitaev}). $c$ is called classical communication cost of the one-shot reverse Shannon simulation.
\end{definition}

\begin{definition}
Let $\cE:\cL(\cH_{A})\rightarrow\cL(\cH_{B})$ be a CPTP map. An \textit{asymptotic reverse Shannon simulation for $\cE$} is a sequence of one-shot reverse Shannon protocols $\cP^{n}$ for $\cE^{\otimes n}$ with error $\eps_{n}$, such that $\lim_{n\rightarrow\infty}\eps_{n}=0$. The classical communication cost $c_{n}$ of this simulation is $\limsup_{n\rightarrow\infty}\frac{c_{n}}{n}=c$.
\end{definition}

\noindent
A precise statement of the Quantum Reverse Shannon Theorem is now as follows.

\begin{theorem}
Let $\cE_{A\rightarrow B}:\cL(\cH_{A})\rightarrow\cL(\cH_{B})$ be a CPTP map. Then the minimal classical communication cost $C_{\mathrm{QRST}}$ of asymptotic reverse Shannon simulations for $\cE_{A\rightarrow B}$ is equal to the entanglement assisted classical capacity $C_{E}$ of $\cE_{A\rightarrow B}$. That is
\begin{align}
C_{\mathrm{QRST}}=\max_{\Phi}I(B:R)_{(\cE\otimes\cI)(\Phi)}\ ,
\end{align}
where $\Phi_{AR}=\proj{\Phi}_{AR}\in\cS_{=}(\cH_{AR})$ is a purification of the input state $\rho_{A}\in\cS_{=}(\cH_{A})$.\footnote{Since all purifications give the same amount of entropy, we do not need to specify which one we use.}
\end{theorem}

\begin{proof}
First note that by the entanglement assisted classical capacity theorem $C_{\mathrm{QRST}}\geq C_{E}$~\cite{Bennett02}.\footnote{Assume that $C_{\mathrm{QRST}}\leq C_{E}-\delta$ for some $\delta>0$ and start with a perfect identity channel $\cI_{A\rightarrow B}$. Then we could use $C_{\mathrm{QRST}}\leq C_{E}-\delta$ together with the entanglement assisted classical capacity theorem to asymptotically simulate the perfect identity channel at a rate $\frac{C_{E}}{C_{E}-\delta}>1$; a contradiction to Holevo's theorem~\cite{NieChu00Book, Holevo98}.} Hence it remains to show that $C_{\mathrm{QRST}}\leq C_{E}$.\\

\noindent
We start by making some general statements about the structure of the proof, and then dive into the technical arguments.\\

\noindent
Because the Quantum Reverse Shannon Theorem makes an asymptotic statement, we have to make our considerations for a general $n\in\mathbb{N}$. Thus the goal is to show the existence of a one-shot reverse Shannon simulation $\cP^{n}_{A\rightarrow B}$ for $\cE_{A\rightarrow B}^{\otimes n}$ that is arbitrarily close to $\cE_{A\rightarrow B}^{\otimes n}$ for $n\rightarrow\infty$, has a classical communication rate of $C_{E}$ and works for any input. We do this by using Quantum State Splitting with embezzling states (Theorem~\ref{thi}), quantum teleportation~\cite{teleport} and the Post-Selection Technique (Proposition~\ref{posti}).\\

\noindent
Any hypothetical map $\cP^{n}_{A\rightarrow B}$ (that we may want to use for the simulation of $\cE_{A\rightarrow B}^{\otimes n}$), can be made to act symmetrically on the $n$-partite input system $\cH_{A}^{\otimes n}$ by inserting a symmetrization step. This works as follows. First Alice and Bob generate some shared randomness by generating maximally entangled states from the embezzling states and measuring their part in the same computational basis (for $n$ large, $O(n\log n)$ maximally entangled states are needed). Then, before the original map $\cP_{A\rightarrow B}^{n}$ starts, Alice applies a random permutation $\pi$ on the input system chosen according to the shared randomness. Afterwards they run the map $\cP_{A\rightarrow B}^{n}$ and then, in the end, Bob undoes the permutation by applying $\pi^{-1}$ on the output system. From this we obtain a permutation invariant version of $\cP^{n}_{A\rightarrow B}$. Since the maximally entangled states can only be created with finite precision, the shared randomness, and therefore the permutation invariance, is not perfect. However, as we will argue at the end, this imperfection can be made arbitrarily small and can therefore be neglected.\\

\noindent
Note that the simulation will need embezzling states $\mu_{A_{\mathrm{emb}}B_{\mathrm{emb}}}$ and maximally entangled states $\proj{\phi_{m}}_{A_{\mathrm{ebit}}B_{\mathrm{ebit}}}$ (for the quantum teleportation step and to assure the permutation invariance). But since the input on these registers is fixed, we are allowed to think of the simulation as a map $\cP_{A\rightarrow B}^{n}$, see Figure~\ref{B}.\\

\noindent
\begin{figure}[ht]
\includegraphics[width=0.50\linewidth]{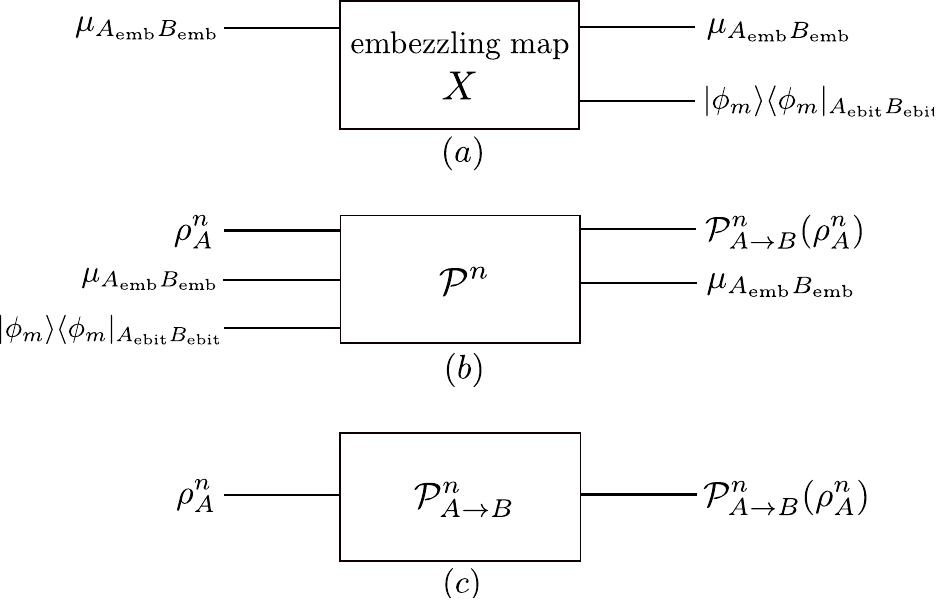}
\caption{(a) $X$ is the map that embezzles $m$ maximally entangled states $\proj{\phi_{m}}_{A_{\mathrm{ebit}}B_{\mathrm{ebit}}}$ out of $\mu_{A_{\mathrm{emb}}B_{\mathrm{emb}}}$. These maximally entangled states are then used in the protocol. (b) The whole map that should simulate $\cE_{A\rightarrow B}^{\otimes n}$ takes $\rho_{A}^{n}\otimes\mu_{A_{\mathrm{emb}}B_{\mathrm{emb}}}\otimes\proj{\phi_{m}}_{A_{\mathrm{ebit}}B_{\mathrm{ebit}}}$ with $\rho_{A}^{n}\in\cS_{=}(\cH_{A}^{\otimes n})$ as an input. But since this input is constant on all registers except for $A$, we can think of the map as in (c), namely as a CPTP map $\cP^{n}_{A\rightarrow B}$ which takes only the input $\rho_{A}^{n}$.}
\label{B}
\end{figure}

\noindent
Let $\beta>0$. Our aim is to show the existence of a map $\cP^{n}_{A\rightarrow B}$, that consists of applying local operations at Alice's side, local operation at Bob's side, sending classical bits from Alice to Bob at a rate of $C_{E}$, and such that
\begin{align}
\|\cE_{A\rightarrow B}^{\otimes n}-\cP_{A\rightarrow B}^{n}\|_{\Diamond}\leq\beta\ .
\label{end}
\end{align}
Because we assume that the map $\cP^{n}_{A\rightarrow B}$ is permutation invariant, we are allowed to use the Post-Selection Technique (Proposition~\ref{posti}). Thus~\eqref{end} relaxes to
\begin{align}
\left\|\left((\cE_{A\rightarrow B}^{\otimes n}-\cP^{n}_{A\rightarrow B})\otimes\cI_{RR'})(\zeta^{n}_{ARR'}\right)\right\|_{1}\leq\beta(n+1)^{-(|A|^{2}-1)}\ ,
\label{bieri}
\end{align}
where $\zeta^{n}_{ARR'}$ is a purification of $\zeta^{n}_{AR}=\int\omega_{AR}^{\otimes n}d(\omega_{AR})$, $\omega_{AR}=\proj{\omega}_{AR}\in\cS_{=}(\cH_{AR})$ and $d(.)$ is the measure on the normalized pure states on $\cH_{AR}$ induced by the Haar measure on the unitary group acting on $\cH_{AR}$, normalized to $\int d(.)=1$.\\

\noindent
To show~\eqref{bieri}, we consider a local simulation of the channel $\cE_{A\rightarrow B}^{\otimes n}$ at Alice's side (using Stinespring's dilation as in~\eqref{intuition}) followed by Quantum State Splitting with embezzling states. Applied to the de Finetti type input state $\zeta^{n}_{ARR'}$, we obtain the state
\begin{align*}
\zeta_{BCRR'}^{n}=(U_{A\rightarrow BC}^{n}\otimes\1_{RR'})\zeta^{n}_{ARR'}(U_{A\rightarrow BC}^{n}\otimes\1_{RR'})^{\dagger}\ .
\end{align*}
As described above, this map can be made permutation invariant (cf.~Figure~\ref{D}).\\

\noindent
\begin{figure}[ht]
\includegraphics[width=0.8\linewidth]{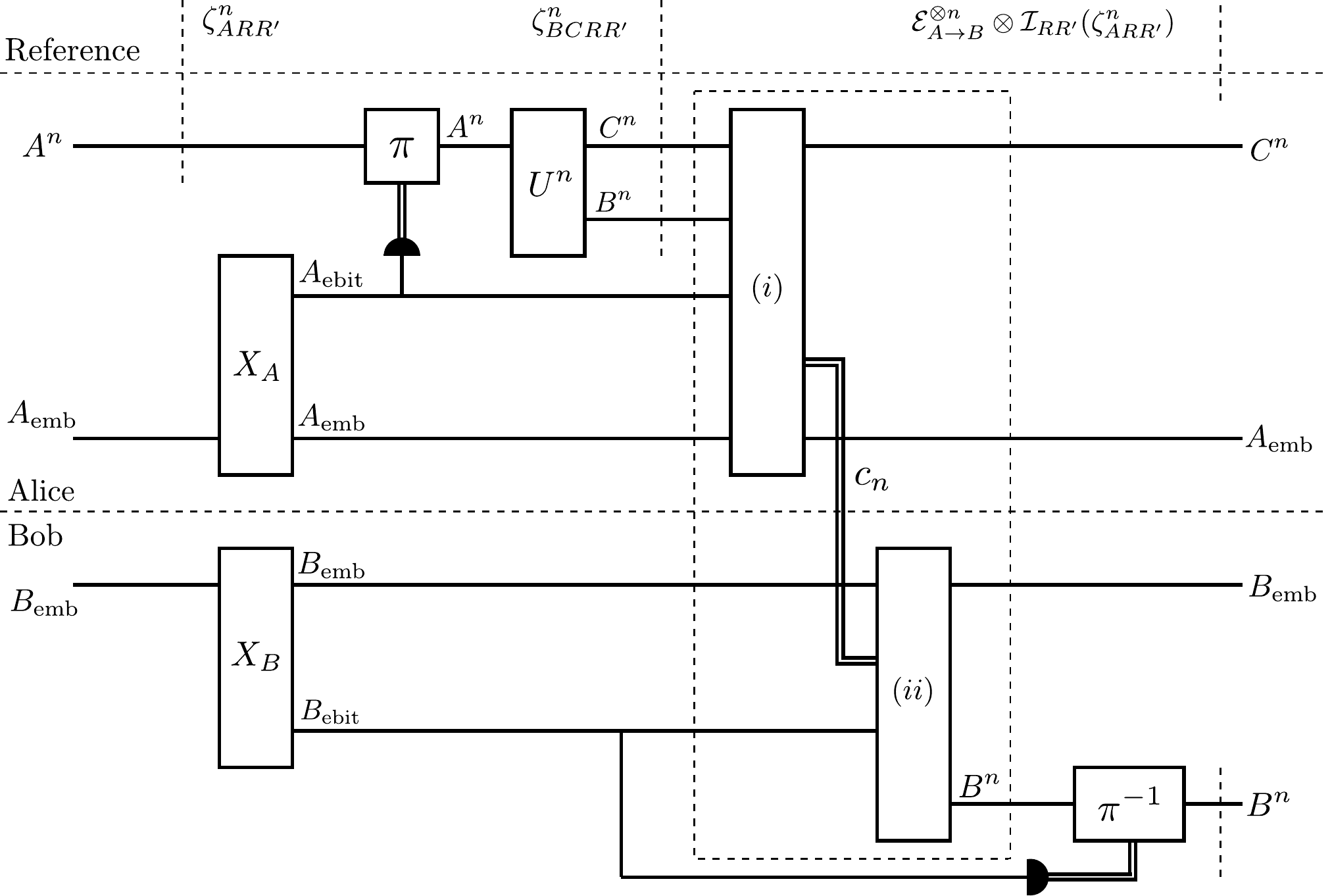}
\caption{A schematic description of the protocol that is used to prove the Quantum Reverse Shannon Theorem. The channel simulation is done for the de Finetti type input state $\zeta^{n}_{ARR'}$. Because our simulation is permutation invariant, the Post-Selection Technique (Proposition~\ref{posti}) shows that this is also sufficient for all input states. The whole simulation is called $\cP^{n}$ in the text. (i) and (ii) denote the subroutine of Quantum State Splitting with embezzling states and quantum teleportation; with local operations on Alice's and Bob's side and a classical communication rate of $c_{n}$.}
\label{D}
\end{figure}

\noindent
Now we use this map as $(\cP^{n}_{A\rightarrow B}\otimes\cI_{RR'})$ in \eqref{bieri}.\footnote{So far this map needs quantum communication, but we are going to replace this by classical communication shortly.} We obtain from the achievability of Quantum State Splitting with embezzling states (Theorem~\ref{thi}) that
\begin{align*}
P\left((\cE_{A\rightarrow B}^{\otimes n}\otimes\cI_{RR'})(\zeta^{n}_{ARR'}),(\cP^{n}_{A\rightarrow B}\otimes\cI_{RR'})(\zeta^{n}_{ARR'})\right)\leq\eps+\eps'+\delta n\cdot\log|B|+|B|^{-n/2}\ ,
\end{align*}
for a quantum communication cost of
\begin{align}
q_{n}\leq\frac{1}{2}I_{\max}^{\eps'}(B:RR')_{(\cE^{\otimes n}\otimes\cI)(\zeta^{n})}+2\cdot\log\frac{1}{\eps}+4+\log n+\log\log|B|\ .
\label{iop}
\end{align}
Because the trace distance is upper bounded by two times the purified distance (Lemma~\ref{a:1}), this implies
\begin{align*}
\left\|\left((\cE_{A\rightarrow B}^{\otimes n}-\cP^{n}_{A\rightarrow B})\otimes\cI_{RR'})(\zeta^{n}_{ARR'}\right)\right\|_{1}\leq2(\eps+\eps'+\delta n\cdot\log|B|+|B|^{-n/2})\ .
\end{align*}
By choosing $\eps=\eps'$ and $\delta=\frac{\eps'}{n\cdot\log|B|}$ we obtain
\begin{align*}
\|((\cE_{A\rightarrow B}^{\otimes n}-\cP^{n}_{A\rightarrow B})\otimes\cI_{RR'})(\zeta^{n}_{ARR'})\|_{1}\leq6\eps'+2\cdot|B|^{-n/2}\ .
\end{align*}
Furthermore we choose $\eps'=\frac{1}{6}\beta(n+1)^{-(|A|^{2}-1)}-\frac{1}{3}|B|^{-n/2}$ (for large enough $n$) and hence
\begin{align*}
\left\|\left((\cE_{A\rightarrow B}^{\otimes n}-\cP^{n}_{A\rightarrow B})\otimes\cI_{RR'})(\zeta^{n}_{ARR'}\right)\right\|_{1}\leq\beta(n+1)^{-(|A|^{2}-1)}\ .
\end{align*}
This is~\eqref{bieri} and by the Post-Selection Technique (Proposition~\ref{posti}) this implies~\eqref{end}.\\

\noindent
But the map $(\cP^{n}_{A\rightarrow B}\otimes\cI_{RR'})$ uses quantum communication and we are only allowed to use classical communication. It thus remains to replace the quantum communication by classical communication and to show that the classical communication rate of the resulting map is upper bounded by $C_{E}$.\\

\noindent
Set $\chi=2\cdot\log\frac{1}{\eps'}+4+\log n+\log\log|B|$. It follows from (\ref{iop}) and below that the quantum communication cost of $(\cP^{n}_{A\rightarrow B}\otimes\cI_{RR'})$ is quantified by
\begin{align*}
q_{n}\leq\frac{1}{2}I_{\max}^{\eps'}(B:RR')_{(\cE^{\otimes n}\otimes\cI)(\zeta^{n})}+\chi\ .
\end{align*}
We can use quantum teleportation~\cite{teleport} (using the maximally entangled states $\proj{\phi_{m}}_{A_{\mathrm{ebit}}B_{\mathrm{ebit}}}$) to transform this into a classical communication cost of
\begin{align*}
c_{n}\leq I_{\max}^{\eps'}(B:RR')_{(\cE^{\otimes n}\otimes\cI)(\zeta^{n})}+2\chi\ .
\end{align*}
By the upper bound in Proposition~\ref{qw} and the fact that we can assume $|R'|\leq(n+1)^{|A|^{2}-1}$ (Proposition~\ref{posti}), we get
\begin{align*}
c_{n} \leq I_{\max}^{\eps'}(B:R)_{(\cE^{\otimes n}\otimes\cI)(\zeta^{n})}+2\cdot\log|R'|+2\chi\leq I_{\max}^{\eps'}(B:R)_{(\cE^{\otimes n}\otimes\cI)(\zeta^{n})}+2\cdot\log\left[(n+1)^{|A|^{2}-1}\right]+2\chi\ .
\end{align*}
By a corollary of Carath\'eodory's theorem (Corollary~\ref{mario}), we can write
$$\zeta_{AR}^{n}=\sum_{i}p_{i}(\Phi^{i}_{AR})^{\otimes n}\ ,$$
where $\Phi^{i}_{AR}=\proj{\Phi^{i}}_{AR}\in\cS_{=}(\cH_{AR})$, $i\in\{1,2,\ldots,(n+1)^{2|A||R|-2}\}$ and $p_{i}$ a probability distribution. Using a quasi-convexity property of the smooth max-information (Proposition~\ref{ach}) we then obtain
\begin{align*}
c_{n} &\leq I_{\max}^{\eps'}(B:R)_{(\cE^{\otimes n}\otimes\cI)(\sum_{i}p_{i}(\Phi^{i})^{\otimes n})}+2\cdot\log\left[(n+1)^{|A|^{2}-1}\right]+2\chi\\
&\leq\max_{i}I_{\max}^{\eps'}(B:R)_{\left[(\cE\otimes\cI)(\Phi^{i})\right]^{\otimes n}}+\log\left[(n+1)^{2|A||R|-2}\right]+2\cdot\log\left[(n+1)^{|A|^{2}-1}\right]+2\chi\\
&\leq\max_{\Phi}I_{\max}^{\eps'}(B:R)_{\left[(\cE\otimes\cI)(\Phi)\right]^{\otimes n}}+\log\left[(n+1)^{2|A||R|-2}\right]+2\cdot\log\left[(n+1)^{|A|^{2}-1}\right]+2\chi\ ,
\end{align*}
where the last maximum ranges over all $\Phi_{AR}=\proj{\Phi}_{AR}\in\cS_{=}(\cH_{AR})$.\\

\noindent
From the Asymptotic Equipartition Property for the smooth max-information (Lemma~\ref{haus}) we obtain
\begin{align*}
c_{n}\leq n\cdot\max_{\Phi}I(B:R)_{(\cE\otimes\cI)(\Phi)}+\sqrt{n}\cdot\xi(\eps')-2\cdot\log\frac{\eps'^{2}}{24}+\log\left[(n+1)^{2|A||R|-2}\right]+2\cdot\log\left[(n+1)^{|A|^{2}-1}\right]+2\chi\ ,
\end{align*}
where $\xi(\eps')=8\sqrt{13-4\cdot\log\eps'}\cdot(2+\frac{1}{2}\cdot\log|A|)$. Since $\eps'=\frac{1}{6}\beta(n+1)^{-(|A|^{2}-1)}-\frac{1}{3}|B|^{-n/2}$, the classical communication rate is then upper bounded by
\begin{align*}
c & =\limsup_{\beta\rightarrow0}\limsup_{n\rightarrow\infty}\frac{c_{n}}{n}\leq\max_{\Phi}I(B:R)_{(\cE\otimes\cI)(\Phi)}\ .
\end{align*}
Thus it only remains to justify why it is sufficient that the maximally entangled states, which we used for the quantum teleportation step and to make the protocol permutation invariant, only have finite precision. For this, it is useful to think of the CPTP map $\cP^{n}_{A\rightarrow B}$ that we constructed above, as in Figure~\ref{B} (b). Let $\eps''>0$ and assume that the entanglement is $\eps''$-close to the perfect input state $\mu_{A_{\mathrm{emb}}B_{\mathrm{emb}}}\otimes\proj{\phi_{m}}_{A_{\mathrm{ebit}}B_{\mathrm{ebit}}}$. The purified distance is monotone (Lemma~\ref{hzhz}) and hence the corresponding imperfect output state is $\eps''$-close to the state obtained under the assumption of perfect permutation invariance. Since $\eps'''$ can be made arbitrarily small (Definition~\ref{p}), the CPTP map based on the imperfect entanglement does the job.

\end{proof}

\section*{Acknowledgments}
We thank J\"urg Wullschleger and Andreas Winter for inspiring discussions and William Matthews and Debbie Leung for detailed feedback on the first version of this paper as well as for suggesting Figures~\ref{C} and~\ref{D}. MB and MC are supported by the Swiss National Science Foundation (grant PP00P2-128455) and the German Science Foundation (grants CH 843/1-1 and CH 843/2-1). RR acknowledges support from the Swiss National Science Foundation (grant No.~200021-119868). Part of this work was carried out while MB and MC were affiliated with the Faculty of Physics at the University of Munich in Germany.

\appendix

\section{Properties of the Purified Distance} \label{app:purdist}

The following gives lower and upper bounds to the puriÞed distance in terms of the trace distance.

\begin{lemma}\cite[Lemma 6]{Tomamichel09}
Let $\rho$, $\sigma\in\cS_{\leq}(\cH)$. Then
\begin{align}
\frac{1}{2}\cdot\|\rho-\sigma\|_{1}\leq P(\rho,\sigma)\leq\sqrt{\|\rho-\sigma\|_{1}+|\tr[\rho]-\tr[\sigma]|}\ .
\end{align}
\label{a:1}
\end{lemma}

\noindent
The purified distance is monotone under CPTP maps.

\begin{lemma}\cite[Lemma 7]{Tomamichel09}
Let $\rho$, $\sigma\in\cS_{\leq}(\cH)$ and $\cE$ be a CPTP map on $\cH$. Then
\begin{align}
P\left(\cE(\rho),\cE(\sigma)\right)\leq P(\rho,\sigma)\ .
\end{align}
\label{hzhz}
\end{lemma}

\noindent
The purified distance is convex in its arguments in the following sense.

\begin{lemma}
Let $\rho_{i}$, $\sigma_{i}\in\cS_{\leq}(\cH)$ with $\rho_{i}\approx_{\eps}\sigma_{i}$ for $i\in I$ and $p_{i}$ a probability distribution. Then
\begin{align}
\sum_{i\in I}p_{i}\rho_{i}\approx_{\eps}\sum_{i\in I}p_{i}\sigma_{i}\ .
\end{align}
\label{333}
\end{lemma}

\begin{proof}
Let $\rho=\sum_{i\in I}p_{i}\rho_{i}$ and $\sigma=\sum_{i\in I}p_{i}\sigma_{i}$ and define $\hat{\rho}=\rho\oplus\left(1-\tr\left[\rho\right]\right)$, $\hat{\sigma}=\rho\oplus\left(1-\tr\left[\sigma\right]\right)$ as well as $\hat{\rho}_{i}=\rho_{i}\oplus\left(1-\tr\left[\rho_{i}\right]\right)$ and $\hat{\sigma}_{i}=\sigma_{i}\oplus\left(1-\tr\left[\sigma_{i}\right]\right)$ for all $i\in I$.\\

\noindent
By assumption we have $F\left(\hat{\rho}_{i},\hat{\sigma}_{i}\right)\geq\sqrt{1-\eps^{2}}$ for all $i\in I$ and using the joint concavity of the fidelity~\cite{NieChu00Book} we obtain
\begin{align*}
P(\rho,\sigma)=\sqrt{1-F^{2}\left(\hat{\rho},\hat{\sigma}\right)}=\sqrt{1-F^{2}\left(\sum_{i\in I}p_{i}\hat{\rho}_{i},\sum_{i\in I}p_{i}\hat{\sigma}_{i}\right)}\leq\sqrt{1-\left(\sum_{i\in I}p_{i}F\left(\hat{\rho}_{i},\hat{\sigma}_{i}\right)\right)^{2}}\leq\eps\ .
\end{align*}
\end{proof}

\section{Basic Properties of Smooth Entropy Measures} \label{app:entropy}

\subsection{Additional Definitions}

Our technical claims use some auxiliary entropic quantities. For $\rho_{A}\in\cS_{\leq}(\cH_{A})$ we define
\begin{align}
&H_{\max}(A)_{\rho}=2\cdot\log\tr\left[\rho_{A}^{1/2}\right]\\
&H_{0}(A)_{\rho}=\log\mathrm{rank}(\rho_{A})\\
&H_{R}(A)_{\rho}=-\sup\left\{\lambda\in\mathbb{R}:\rho_{A}\geq2^{\lambda}\cdot\rho_{A}^{0}\right\}\ ,
\end{align}
where $H_{\max}(A)_{\rho}$ is called \emph{max-entropy}. For $\eps\geq0$, the \emph{smooth max-entropy of $\rho_{A}\in\cS_{\leq}(\cH_{A})$} is defined as
\begin{align}
H_{\max}^{\eps}(A)_{\rho}=\inf_{\bar{\rho}_{A}\in\cB^{\eps}(\rho_{A})}H_{\max}(A)_{\bar{\rho}}\ .
\end{align}

\noindent
The \textit{conditional min-entropy of $\rho_{AB}\in\cS_{\leq}(\cH_{AB})$ relative to $\sigma_{B}\in\cS_{=}(\cH_{B})$} is given by
\begin{align}
H_{\min}(A|B)_{\rho|\sigma}=-D_{\max}(\rho_{AB}\|\1_{A}\otimes\sigma_{B})\ .
\end{align}

\noindent
The \emph{quantum conditional collision entropy of A given B} for $\rho_{AB}\in\cS_{\leq}(\cH_{AB})$ is defined as
\begin{align}
H_{C}(A|B)_{\rho}=-\inf_{\sigma_{B}\in\cS_{=}(\cH_{B})}\log\tr\left[\left((\1_{A}\otimes\sigma_{B}^{-1/4})\rho_{AB}(\1_{A}\otimes\sigma_{B}^{-1/4})\right)^{2}\right]\ ,
\end{align}
where the inverses are generalized inverses.\footnote{For $\rho\in\cP(\cH)$, $\rho^{-1}$ is a generalized inverse of $\rho$ if $\rho\rho^{-1}=\rho^{-1}\rho=\rho^{0}=(\rho^{-1})^{0}$.}\\

\noindent
As in~\cite{datta-2008-2} we define the \emph{min-relative entropy of $\rho\in\cS_{\leq}(\cH)$ with respect to $\sigma\in\cP(\cH)$} as
\begin{align}
D_{\min}(\rho\|\sigma)=-\log\tr\left[\rho^{0}\sigma\right]\ .
\end{align}

\subsection{Alternative Formulas}

The max-relative entropy can be written in the following alternative form.

\begin{lemma}\cite[Lemma B.5.3]{Ren05}
Let $\rho\in\cS_{\leq}(\cH)$ and $\sigma\in\cP(\cH)$ such that $\mathrm{supp}(\rho)\subseteq\mathrm{supp}(\sigma)$. Then
\begin{align}
D_{\max}(\rho\|\sigma)=\log\|\sigma^{-1/2}\rho\sigma^{-1/2}\|_{\infty}
\end{align}
where the inverses are generalized inverses.
\label{inverse}
\end{lemma}

\noindent
Using this we can give an alternative expression for the max-information.

\begin{lemma}
Let $\rho_{AB}\in\cS_{\leq}(\cH_{AB})$. Then
\begin{align}
I_{\max}(A:B)_{\rho}=H_{0}(A)_{\rho}-H_{\min}(A|B)_{\rho_{B|A}}
\end{align}
where $\rho_{B|A}=(\rho_{A}\otimes\1_{B})^{-1/2}\frac{\rho_{AB}}{\mathrm{rank}(\rho_{A})}(\rho_{A}\otimes\1_{B})^{-1/2}$ and the inverses are generalized inverses.
\label{new}
\end{lemma}

\begin{proof}
Without loss of generality we can restrict the minimum in the definition of the max-mutual information to $\sigma_{B}\in\cS_{=}(\cH_{B})$ with $\mathrm{supp}(\rho_{AB})\subseteq\mathrm{supp}(\rho_{A})\otimes\mathrm{supp}(\sigma_{B})$. To see this note that $D_{\max}(\rho_{AB}\|\rho_{A}\otimes\rho_{B})$ is finite but $D_{\max}(\rho_{AB}\|\rho_{A}\otimes\sigma_{B})=\infty$ for any $\sigma_{B}\in\cS_{=}(\cH_{B})$ with $\mathrm{supp}(\rho_{AB})\nsubseteq\mathrm{supp}(\rho_{A})\otimes\mathrm{supp}(\sigma_{B})$.\\

\noindent
Therefore we can use Lemma~\ref{inverse}
\begin{align*}
I_{\max}(A:B)_{\rho} & =\min_{\sigma_{B}}\log\left\|(\rho_{A}\otimes\sigma_{B})^{-1/2}\rho_{AB}(\rho_{A}\otimes\sigma_{B})^{-1/2}\right\|_{\infty}\\
& =\min_{\sigma_{B}}\log\left\|(\frac{\1_{A}}{\mathrm{rank}(\rho_{A})}\otimes\sigma_{B})^{-1/2}(\rho_{A}\otimes\1_{B})^{-1/2}\frac{\rho_{AB}}{\mathrm{rank}(\rho_{A})}(\rho_{A}\otimes\1_{B})^{-1/2}(\frac{\1_{A}}{\mathrm{rank}(\rho_{A})}\otimes\sigma_{B})^{-1/2}\right\|_{\infty}\ .
\end{align*}
We have $\rho_{B|A}=(\rho_{A}\otimes\1_{B})^{-1/2}\frac{\rho_{AB}}{\mathrm{rank}(\rho_{A})}(\rho_{A}\otimes\1_{B})^{-1/2}\in\cS_{=}(\cH_{AB})$ because
\begin{align*}
\tr\left[\rho_{B|A}\right]=\tr\left[\frac{(\rho_{A}^{-1}\otimes\1_{B})\rho_{AB}}{\mathrm{rank}(\rho_{A})}\right]=\tr\left[\frac{\rho_{A}^{-1}\rho_{A}}{\mathrm{rank}(\rho_{A})}\right]=1\ .
\end{align*}
Hence we can write
\begin{align*}
I_{\max}(A:B)_{\rho} =\min_{\sigma_{B}}D_{\max}(\rho_{B|A}\|\frac{\1_{A}}{\mathrm{rank}(\rho_{A})}\otimes\sigma_{B})=H_{0}(A)_{\rho}-H_{\min}(A|B)_{\rho_{B|A}}\ .
\end{align*}
\end{proof}

\subsection{Upper and Lower Bounds}

The conditional min-entropy is upper bounded by the quantum conditional collision entropy.

\begin{lemma}
Let $\rho_{AB}\in\cS_{\leq}(\cH_{AB})$. Then
\begin{align}
H_{\min}(A|B)_{\rho}\leq H_{C}(A|B)_{\rho}\ .
\end{align}
\label{collision}
\end{lemma}

\begin{proof}
Let $\sigma_{B}\in\cS_{=}(\cH_{B})$ be such that $H_{\min}(A|B)_{\rho}=-D_{\max}(\rho_{AB}\|\rho_{A}\otimes\sigma_{B})$. We know that $\mathrm{supp}(\rho_{AB})\subseteq\1_{A}\otimes\mathrm{supp}(\sigma_{B})$ (argumentation analogue as in the proof of Lemma~\ref{new}) and hence we can use the alternative expression for the max-relative entropy (Lemma~\ref{inverse})
\begin{align*}
H_{\min}(A|B)_{\rho}=-\log\max_{\omega_{AB}\in\cS_{=}(\cH_{AB})}\tr\left[\omega_{AB}\left(\1_{A}\otimes\sigma_{B}^{-1/2}\right)\rho_{AB}\left(\1_{A}\otimes\sigma_{B}^{-1/2}\right)\right]\ .
\end{align*}
But for $\hat{\rho}_{AB}=\frac{\rho_{AB}}{\tr\left[\rho_{AB}\right]}\in\cS_{=}(\cH_{AB})$ we have
\begin{align*}
H_{C}(A|B)_{\rho}&=-\log\min_{\kappa_{B}\in\cS_{=}(\cH_{B})}\tr\left[\rho_{AB}\left(\1_{A}\otimes\kappa_{B}^{-1/2}\right)\rho_{AB}\left(\1_{A}\otimes\kappa_{B}^{-1/2}\right)\right]\\
&\geq-\log\tr\left[\rho_{AB}\left(\1_{A}\otimes\sigma_{B}^{-1/2}\right)\rho_{AB}\left(\1_{A}\otimes\sigma_{B}^{-1/2}\right)\right]\\
&=-\log\tr\left[\rho_{AB}\right]-\log\tr\left[\hat{\rho}_{AB}\left(\1_{A}\otimes\sigma_{B}^{-1/2}\right)\rho_{AB}\left(\1_{A}\otimes\sigma_{B}^{-1/2}\right)\right]\\
&\geq-\log\max_{\omega_{AB}\in\cS_{=}(\cH_{AB})}\tr\left[\omega_{AB}\left(\1_{A}\otimes\sigma_{B}^{-1/2}\right)\rho_{AB}\left(\1_{A}\otimes\sigma_{B}^{-1/2}\right)\right]=H_{\min}(A|B)_{\rho}\ .
\end{align*}
\end{proof}

\noindent
The max-relative entropy is lower bounded by the quantum relative entropy.

\begin{lemma}\cite[Lemma 10]{datta-2008-2}
Let $\rho$, $\sigma\in\cS_{\leq}(\cH)$. Then
\begin{align}
D_{\max}(\rho\|\sigma)\geq D(\rho\|\sigma)\ .
\end{align}
\label{relation}
\end{lemma}

\noindent
The min-relative entropy is nonnegative for normalized states.

\begin{lemma}\cite[Lemma 6]{datta-2008-2}.
Let $\rho$, $\sigma\in\cS_{=}(\cH)$. Then $D_{\min}(\rho\|\sigma)\geq0$.
\label{min-relative}
\end{lemma}

\noindent
From this we find the following dimension lower bounds for the conditional min-entropy.

\begin{lemma}
Let $\rho_{AB}\in\cS_{\leq}(\cH_{AB})$. Then
\begin{align}
& -\log|A|\leq H_{\min}(A|B)_{\rho|\rho}\\
& -\log|A|\leq H_{\min}(A|B)_{\rho}+\log\tr[\rho_{AB}]\ .
\label{uj}
\end{align}
\label{cond-min}
\end{lemma}

\begin{proof}
Let $\rho_{ABR}\in\cS_{\leq}(\cH_{ABR})$ be a purification of $\rho_{AB}$ and define $\hat{\rho}_{ABR}=\frac{\rho_{ABR}}{\tr[\rho_{ABR}]}\in\cS_{=}(\cH_{ABR})$. By a duality property of min- and max-entropy~\cite[Proposition 3.10]{Berta08} we have
\begin{align*}
H_{\min}(A|B)_{\hat{\rho}|\hat{\rho}}=\min_{\sigma_{R}\in\cS_{=}(\cH_{R})}D_{\min}(\hat{\rho}_{AR}\|\1_{A}\otimes\sigma_{R})\ .
\end{align*}
Using the nonnegativity of the min-relative entropy for normalized states (Lemma~\ref{min-relative}) we obtain
\begin{align*}
0\leq\min_{\sigma_{R}\in\cS_{=}(\cH_{R})}D_{\min}\left(\hat{\rho}_{AR}\|\frac{\1_{A}}{|A|}\otimes\sigma_{R}\right)&=\log|A|+\min_{\sigma_{R}\in\cS_{=}(\cH_{R})}D_{\min}\left(\hat{\rho}_{AR}\|\1_{A}\otimes\sigma_{R}\right)=\log|A|+H_{\min}(A|B)_{\hat{\rho}|\hat{\rho}}\\
&=\log|A|+H_{\min}(A|B)_{\rho|\rho}\ .
\end{align*}
Inequality (\ref{uj}) is proved in~\cite[Lemma 20]{Tomamichel09}.
\end{proof}

\noindent
The following dimension upper bound holds for the max-information.

\begin{lemma}
Let $\rho_{AB}\in\cS_{\leq}(\cH_{AB})$. Then
\begin{align}
I_{\max}(A:B)_{\rho}\leq2\cdot\log\min\left\{|A|,|B|\right\}\ .
\end{align}
\label{upperb}
\end{lemma}

\begin{proof}
It follows from the dimension lower bound for the conditional min-entropy (Lemma~\ref{cond-min}) that
\begin{align*}
I_{\max}(A:B)_{\rho}\leq D_{\max}\left(\rho_{AB}\|\rho_{A}\otimes\frac{\1_{B}}{|B|}\right)=\log|B|-H_{\min}(B|A)_{\rho|\rho}\leq2\cdot\log|B|\ .
\end{align*}
Using the alternative expression for the max-information (Lemma~\ref{new}) and again Lemma~\ref{cond-min} we get
\begin{align*}
I_{\max}(A:B)_{\rho}=H_{0}(A)_{\rho}-H_{\min}(A|B)_{\rho_{B|A}}\leq2\cdot\log|A|\ .
\end{align*}
\end{proof}

\begin{remark}
In general there is no dimension upper bound for $D_{\max}(\rho_{AB}\|\rho_{A}\otimes\rho_{B})$ as can be seen by the following example. 
Let $\rho_{AB}=\ket{\rho}\bra{\rho}_{AB}\in\cS_{=}(\cH_{AB})$ with Schmidt-decomposition $\ket{\rho}_{AB}=\sum_{i}\lambda_{i}\ket{i}_{A}\otimes\ket{i}_{B}$. Then
\begin{align}
D_{\max}(\rho_{AB}\|\rho_{A}\otimes\rho_{B})=\log\left(\sum_{i}\lambda_{i}^{-1}\right)\ ,
\end{align}
where the sum ranges over all $i$ with $\lambda_{i}>0$.
\end{remark}

\noindent
The following is a bound on the increase of the max-information when an additional subsystem is added.

\begin{lemma}
Let $\eps\geq0$ and $\rho_{ABC}\in\cS_{=}(\cH_{ABC})$. Then
\begin{align}
I_{\max}^{\eps}(A:BC)_{\rho}\leq I_{\max}^{\eps}(A:B)_{\rho}+2\cdot\log|C|\ .
\end{align}
\label{qw}
\end{lemma}

\begin{proof}
Let $\tilde{\rho}_{AB}\in\cB^{\eps}(\rho_{AB})$ and $\tilde{\sigma}_{B}\in\cS_{=}(\cH_{B})$ such that $I_{\max}^{\eps}(A:B)_{\rho}=D_{\max}\left(\tilde{\rho}_{AB}\|\tilde{\rho}_{A}\otimes\tilde{\sigma}_{B}\right)=\log\mu$. That is, $\mu$ is minimal such that $\mu\cdot\tilde{\rho}_{A}\otimes\tilde{\sigma}_{B}\geq\tilde{\rho}_{AB}$ and this implies $\mu\cdot\tilde{\rho}_{A}\otimes\tilde{\sigma}_{B}\otimes\frac{\1_{C}}{|C|}\geq\frac{1}{|C|}\cdot\tilde{\rho}_{AB}\otimes\1_{C}$. Furthermore define $\tilde{\rho}_{ABC}\in\cB^{\eps}(\rho_{ABC})$ such that $\tr_{C}\left[\tilde{\rho}_{ABC}\right]=\tilde{\rho}_{AB}$ (by Uhlmann's theorem such a state exists~\cite{Uhlmann76, Jozsa94}).\\

\noindent
By the dimension lower bound for the min-entropy (Lemma~\ref{cond-min}), we have $H_{\min}(C|AB)_{\tilde{\rho}|\tilde{\rho}}\geq-\log|C|$. Therefore $|C|\cdot\tilde{\rho}_{AB}\otimes\1_{C}\geq\tilde{\rho}_{ABC}$ and hence $\mu\cdot\tilde{\rho}_{A}\otimes\tilde{\sigma}_{B}\otimes\frac{\1_{C}}{|C|}\geq\frac{1}{|C|^{2}}\cdot\tilde{\rho}_{ABC}$.\\

\noindent
Now let $D_{\max}\left(\tilde{\rho}_{ABC}\|\tilde{\rho}_{A}\otimes\tilde{\sigma}_{B}\otimes\frac{\1_{C}}{|C|}\right)=\log\lambda$. That is, $\lambda$ is minimal such that $\lambda\cdot\tilde{\rho}_{A}\otimes\tilde{\sigma}_{B}\otimes\frac{\1_{C}}{|C|}\geq\tilde{\rho}_{ABC}$. Thus it follows that $\lambda\leq\mu\cdot|C|^{2}$ and from this we get
\begin{align*}
I_{\max}^{\eps}(A:BC)_{\rho}\leq D_{\max}\left(\tilde{\rho}_{ABC}\|\tilde{\rho}_{A}\otimes\tilde{\sigma}_{B}\otimes\frac{\1_{C}}{|C|}\right)\leq D_{\max}\left(\tilde{\rho}_{AB}\|\tilde{\rho}_{A}\otimes\tilde{\sigma}_{B}\right)+2\log|C|=I_{\max}^{\eps}(A:B)_{\rho}+2\cdot\log|C|\ .
\end{align*}
\end{proof}

\noindent
The max-information can be lower bounded in terms of the min-entropy.

\begin{lemma}
Let $\rho_{AB}\in\cS_{\leq}(\cH_{AB})$. Then
\begin{align}
I_{\max}(A:B)_{\rho}\geq H_{\min}(A)_{\rho}-H_{\min}(A|B)_{\rho}\ .
\end{align}
\label{tschau2}
\end{lemma}

\begin{proof}
Let $\tilde{\sigma}_{B}\in\cS_{=}(\cH_{B})$ such that $I_{\max}(A:B)_{\rho}=D_{\max}(\rho_{AB}\|\rho_{A}\otimes\tilde{\sigma}_{B})=\log\lambda$. That is, $\lambda$ is minimal such that $\lambda\cdot\rho_{A}\otimes\tilde{\sigma}_{B}\geq\rho_{AB}$. Furthermore let $\mu$ be minimal such that  $\mu\cdot\left\|\rho_{A}\right\|_{\infty}\cdot\1_{A}\otimes\tilde{\sigma}_{B}\geq\rho_{AB}$.\\

\noindent
Since $\|\rho_{A}\|_{\infty}\cdot\1_{A}\geq\rho_{A}$, we have that $\lambda\geq\mu$. Now set $H_{\min}(A|B)_{\rho|\tilde{\sigma}}=-\log\nu$, i.e.~$\nu$ is minimal such that $\nu\cdot\1_{A}\otimes\tilde{\sigma}_{B}\geq\rho_{AB}$. Thus $\nu=\mu\cdot\|\rho_{A}\|_{\infty}$ and we conclude
\begin{align*}
I_{\max}(A:B)_{\rho}=\log\lambda\geq\log\mu=-\log\|\rho_{A}\|_{\infty}+\log\nu=H_{\min}(A)_{\rho}
-H_{\min}(A|B)_{\rho|\tilde{\sigma}}\geq H_{\min}(A)_{\rho}-H_{\min}(A|B)_{\rho}\ .
\end{align*}
\end{proof}

\noindent
The max-information can be upper bounded in terms of a difference between two entropic quantities.

\begin{lemma}
Let $\rho_{AB}\in\cS_{\leq}(\cH_{AB})$. Then
\begin{align}
I_{\max}(A:B)_{\rho}\leq H_{R}(A)_{\rho}-H_{\min}(A|B)_{\rho}\ .
\end{align}
\label{kleber}
\end{lemma}

\begin{proof}
Let $\tilde{\sigma}_{B}\in\cS_{=}(\cH_{B})$ such that $H_{\min}(A|B)_{\rho}=H_{\min}(A|B)_{\rho|\tilde{\sigma}}=-\log\mu$. That is, $\mu$ is minimal such that $\mu\cdot\1_{A}\otimes\tilde{\sigma}_{B}\geq\rho_{AB}$. Since multiplication by $\rho_{A}^{0}\otimes\1_{B}$ does not affect $\rho_{AB}$ (note that the support of $\rho_{AB}$ is contained in the support of $\rho_{A}\otimes\rho_{B}$), we also have $\mu\cdot\rho_{A}^{0}\otimes\tilde{\sigma}_{B}\geq\rho_{AB}$.\\

\noindent
Furthermore $\rho_{A}\geq\lambda_{\min}(\rho_{A})\cdot\rho_{A}^{0}$ where $\lambda_{\min}(\rho)$ denotes the smallest non-zero eigenvalue of $\rho$. Therefore $\frac{\mu}{\lambda_{\min}(\rho_{A})}\cdot\rho_{A}\otimes\tilde{\sigma}_{B}\geq\rho_{AB}$. Now let $D_{\max}(\rho_{AB}|\rho_{A}\otimes\tilde{\sigma}_{B})=\log\lambda$, i.e.~$\lambda$ is minimal such that $\lambda\cdot\rho_{A}\otimes\tilde{\sigma}_{B}\geq\rho_{AB}$. Hence $\lambda\leq\mu\cdot\lambda^{-1}_{\min}(\rho_{A})$ and thus
$$I_{\max}(A:B)_{\rho}\leq D_{\max}(\rho_{AB}|\rho_{A}\otimes\tilde{\sigma}_{B})\leq H_{R}(A)_{\rho}-H_{\min}(A|B)_{\rho}\ .$$
\end{proof}

\noindent
This can be generalized to a version for the smooth max-information.

\begin{lemma}
Let $\eps>0$ and $\rho_{AB}\in\cS_{=}(\cH_{AB})$. Then
\begin{align}
I_{\max}^{\eps}(A:B)_{\rho}\leq H_{\max}^{\eps^{2}/48}(A)_{\rho}-H_{\min}^{\eps^{2}/48}(A|B)_{\rho}-2\cdot\log\frac{\eps^{2}}{24}\ .
\end{align}
\label{hello}
\end{lemma}

\begin{proof}
By the entropy measure upper bound for the max-information (Lemma~\ref{kleber}) we obtain
\begin{align*}
I_{\max}^{\eps}(A:B)_{\rho}&\leq\min_{\bar{\rho}_{AB}\in\cB^{\eps}(\rho_{AB})}\left[H_{R}(A)_{\bar{\rho}}-H_{\min}(A|B)_{\bar{\rho}}\right]\\
&\leq\min_{\omega_{AB}\in\cB^{\eps^{2}/48}(\rho_{AB})}\left\{\min_{\Pi_{A}}\left[H_{R}(A)_{\Pi_{A}\omega\Pi_{A}}-H_{\min}(A|B)_{\Pi_{A}\omega\Pi_{A}}\right]\right\}\ ,
\end{align*}
where the minimum ranges over all $0\leq\Pi_{A}\leq\1_{A}$ such that $\Pi_{A}\omega_{AB}\Pi_{A}\approx_{\eps/2}\omega_{AB}$. Now we choose $\tilde{\sigma}_{B}\in\cS_{=}(\cH_{B})$ such that $H_{\min}(A|B)_{\omega|\tilde{\sigma}}=H_{\min}(A|B)_{\omega}$ and use Lemma~\ref{kkk} to obtain
\begin{align*}
I_{\max}^{\eps}(A:B)_{\rho}&\leq\min_{\omega_{AB}\in\cB^{\eps^{2}/48}(\rho_{AB})}\left\{\min_{\Pi_{A}}\left[H_{R}(A)_{\Pi_{A}\omega\Pi_{A}}-H_{\min}(A|B)_{\Pi_{A}\omega\Pi_{A}|\tilde{\sigma}}\right]\right\}\\
&\leq\min_{\omega_{AB}\in\cB^{\eps^{2}/48}(\rho_{AB})}\left\{\min_{\Pi_{A}}\left[H_{R}(A)_{\Pi_{A}\omega\Pi_{A}}\right]-H_{\min}(A|B)_{\omega|\tilde{\sigma}}\right\}\\
&=\min_{\omega_{AB}\in\cB^{\eps^{2}/48}(\rho_{AB})}\left\{\min_{\Pi_{A}}\left[H_{R}(A)_{\Pi_{A}\omega\Pi_{A}}\right]-H_{\min}(A|B)_{\omega}\right\}\ ,
\end{align*}
where the minimum ranges over all $0\leq\Pi_{A}\leq\1_{A}$ such that $\Pi_{A}\omega_{AB}\Pi_{A}\approx_{\eps/2}\omega_{AB}$. As a next step we choose $\omega_{AB}=\tilde{\omega}_{AB}\in\cB^{\eps^{2}/48}(\rho_{AB})$ such that $H_{\min}^{\eps^{2}/48}(A|B)_{\rho}=H_{\min}(A|B)_{\tilde{\omega}}$. Hence we get
\begin{align*}
I_{\max}^{\eps}(A:B)_{\rho}\leq \min_{\Pi_{A}}\left[H_{R}(A)_{\Pi_{A}\tilde{\omega}\Pi_{A}}\right]-H_{\min}^{\eps^{2}/48}(A|B)_{\rho}\ ,
\end{align*}
where the minimum ranges over all $0\leq\Pi_{A}\leq\1_{A}$ such that $\Pi_{A}\tilde{\omega}_{AB}\Pi_{A}\approx_{\eps/2}\tilde{\omega}_{AB}$.
Using Lemma~\ref{980}, we can choose $0\leq\Pi_{A}\leq\1_{A}$ with $\Pi_{A}\tilde{\omega}_{AB}\Pi_{A}\approx_{\eps/2}\tilde{\omega}_{AB}$ such that
$H_{R}(A)_{\Pi_{A}\tilde{\omega}\Pi_{A}}\leq H_{\max}^{\eps^{2}/24}(A)_{\tilde{\omega}}-2\cdot\log\frac{\eps^{2}}{24}$. From this we finally obtain
\begin{align*}
I_{\max}^{\eps}(A:B)_{\rho}\leq H_{\max}^{\eps^{2}/24}(A)_{\tilde{\omega}}-H_{\min}^{\eps^{2}/48}(A|B)_{\rho}-2\cdot\log\frac{\eps^{2}}{24}
\leq H_{\max}^{\eps^{2}/48}(A)_{\rho}-H_{\min}^{\eps^{2}/48}(A|B)_{\rho}-2\cdot\log\frac{\eps^{2}}{24}\ .
\end{align*}
\end{proof}

\subsection{Monotonicity}

The max-relative entropy is monotone under CPTP maps.

\begin{lemma}\cite[Lemma 7]{datta-2008-2}
Let  $\rho,\sigma\in\cP(\cH)$ and $\cE$ be a CPTP map on $\cH$. Then
\begin{align}
D_{\max}(\rho\|\sigma)\geq D_{\max}(\cE(\rho)\|\cE(\sigma))\ .
\end{align}
\label{lalala}
\end{lemma}

\noindent
It follows that the max-information is monotone under local CPTP maps.

\begin{lemma}
Let $\rho_{AB}\in\cS_{\leq}(\cH_{AB})$ and $\cT$ be a CPTP map on $\cH_{AB}$ of the form $\cE=\cE_{A}\otimes\cE_{B}$. Then
\begin{align}
I_{\max}(A:B)_{\rho}\geq I_{\max}(A:B)_{\cE(\rho)}\ .
\end{align}
\label{epsilon}
\end{lemma}

\begin{proof}
Let $\tilde{\sigma}_{B}\in\cS_{=}(\cH_{B})$. Using the monotonicity of the max-information under local CPTP maps (Lemma~\ref{lalala}) we obtain
\begin{align*}
I_{\max}(A:B)_{\rho}=D_{\max}(\rho_{AB}\|\rho_{A}\otimes\tilde{\sigma}_{B})&\geq D_{\max}(\cE(\rho_{AB})\|\cE_{A}(\rho_{A})\otimes\cE_{B}(\tilde{\sigma}_{B}))\\
&\geq\min_{\omega_{B}\in\cS_{=}(\cH_{B})}D_{\max}(\cE(\rho_{AB})\|\cE_{A}(\rho_{A})\otimes\omega_{B})\\
&=I_{\max}(A:B)_{\cE(\rho)}\ .
\end{align*}
\end{proof}

\subsection{Miscellaneous Properties}

The max-information of classical-quantum states can be estimated as follows.

\begin{lemma}
Let $\rho_{ABI}\in\cS_{\leq}(\cH_{ABI})$ with $\rho_{ABI}=\sum_{i\in I}p_{i}\rho_{AB}^{i}\otimes\proj{i}_{I}$, $\rho_{AB}^{i}\in\cS_{\leq}(\cH_{AB})$ and $p_{i}>0$ for $i\in I$ as well as the $\ket{i}$ mutually orthogonal (i.e.~the state is classical on $I$). Then
\begin{align}
I_{\max}(AI:B)_{\rho}\geq\max_{i\in I}I_{\max}(A:B)_{\rho^{i}}\ .
\end{align}
\label{hierbinichnicht}
\end{lemma}

\begin{proof}
Let $\tilde{\sigma}_{B}\in\cS_{=}(\cH_{B})$ such that $I_{\max}(AI:B)_{\rho}=D_{\max}(\rho_{ABI}\|\rho_{AI}\otimes\tilde{\sigma}_{B})=\log\lambda$. That is, $\lambda$ is minimal such that 
$$\lambda\cdot\sum_{i}p_{i}\rho_{A}^{i}\otimes\tilde{\sigma}_{B}\otimes\proj{i}\geq\sum_{i}p_{i}\rho_{AB}^{i}\otimes\proj{i}\ .$$
Since the $\ket{i}$ are mutually orthogonal and $p_{i}>0$ for $i\in I$, this is equivalent to $\forall i\in I:\lambda\cdot\rho_{A}^{i}\otimes\tilde{\sigma}_{B}\geq\rho_{AB}^{i}$.\\

\noindent
Set $D_{\max}(\rho_{AB}^{i}\|\rho_{A}^{i}\otimes\tilde{\sigma}_{B})=\log\lambda_{i}$, i.e.~$\lambda_{i}$ is minimal such that $\lambda_{i}\cdot\rho_{A}^{i}\otimes\tilde{\sigma}_{B}\geq\rho_{AB}^{i}$. Hence $\lambda\geq\max_{i\in I}\lambda_{i}$ and therefore
\begin{align*}
I_{\max}(AI:B)_{\rho} =\log\lambda\geq\max_{i\in I}\lambda_{i}=\max_{i\in I}D_{\max}(\rho_{AB}^{i}\|\rho_{A}^{i}\otimes\tilde{\sigma}_{B})
\geq\max_{i\in I}I_{\max}(A:B)_{\rho^{i}}\ .
\end{align*}
\end{proof}

\noindent
From this we obtain the following corollary about the behavior of the max-information under projective measurements.

\begin{corollary}
Let $\rho_{AB}\in\cS_{\leq}(\cH_{AB})$ and let $P=\left\{P_{A}^{i}\right\}_{i\in I}$ be a collection of projectors that describe a projective measurement on system $A$. For $\tr\left[P^{i}_{A}\rho_{AB}\right]\neq0$, let $p_{i}=\tr\left[P^{i}_{A}\rho_{AB}\right]$ and $\rho_{AB}^{i}=\frac{1}{p_{i}}P^{i}_{A}\rho_{AB}P^{i}_{A}$. Then
\begin{align}
I_{\max}(A:B)_{\rho}\geq\max_{i}I_{\max}(A:B)_{\rho^{i}}\ ,
\end{align}
where the maximum ranges over all $i$ for which $\rho_{AB}^{i}$ is defined.
\label{hehehehe}
\end{corollary}

\begin{proof}
Define a CPTP map $\cE:\cS_{\leq}(\cH_{AB})\mapsto\cS_{\leq}(\cH_{ABI})$ with $\cE(.)=\sum_{i}\left[P_{A}^{i}(.)P_{A}^{i}\right]\otimes\proj{i}_{I}$,
where the $\ket{i}$ are mutually orthogonal. Then the monotonicity of the max-information under local CPTP maps (Lemma~\ref{epsilon}) combined with the preceding lemma about the max-information of classical-quantum states (Lemma~\ref{hierbinichnicht}) show that
\begin{align*}
I_{\max}(A:B)_{\rho}\geq I_{\max}(A:B)_{\cE(\rho)}\geq\max_{i}I_{\max}(A:B)_{\rho^{i}}\ .
\end{align*}
\end{proof}

\noindent
The max-relative entropy is quasi-convex in the following sense.

\begin{lemma}\cite[Lemma 9]{datta-2008-2}
Let $\rho=\sum_{i\in I}p_{i}\rho_{i}\in\cS_{\leq}(\cH)$ and $\sigma=\sum_{i\in I}p_{i}\sigma_{i}\in\cS_{\leq}(\cH)$ with $\rho_{i},\sigma_{i}\in\cS_{\leq}(\cH)$ for $i\in I$. Then
\begin{align}
D_{\max}(\rho\|\sigma)\leq\max_{i\in I}D_{\max}(\rho_{i}\|\sigma_{i})\ .
\end{align}
\label{quasi6}
\end{lemma}

\noindent
From this we find the following quasi-convexity type lemma for the smooth max-information.

\begin{lemma}
Let $\eps\geq0$ and $\rho_{AB}=\sum_{i\in I}p_{i}\rho_{AB}^{i}\in\cS_{\leq}(\cH_{AB})$ with $\rho_{AB}^{i}\in\cS_{\leq}(\cH_{AB})$ for $i\in I$ . Then
\begin{align}
I_{\max}^{\eps}(A:B)_{\rho}\leq\max_{i\in I}I_{\max}^{\eps}(A:B)_{\rho^{i}}+\log|I|\ .
\end{align}
\label{ach}
\end{lemma}

\begin{proof}
Let $\tilde{\rho}_{AB}^{i}\in\cB^{\epsilon}(\rho_{AB}^{i})$ and $\tilde{\sigma}^{i}_{B}\in\cS_{=}(\cH_{B})$ for $i\in I$. Using the quasi-convexity of the max-relative entropy (Lemma~\ref{quasi6}) we obtain
\begin{align*}
\max_{i\in I}I_{\max}^{\eps}(A:B)_{\rho^{i}}+\log|I|&=\max_{i\in I}D_{\max}\left(\tilde{\rho}_{AB}^{i}\|\tilde{\rho}_{A}^{i}\otimes\tilde{\sigma}_{B}^{i}\right)+\log|I|\geq D_{\max}\left(\sum_{i\in I}p_{i}\tilde{\rho}_{AB}^{i}\|\sum_{i\in I}p_{i}\tilde{\rho}_{A}^{i}\otimes\tilde{\sigma}_{B}^{i}\right)+\log|I|\\
&=\log\min\left\{\lambda\in\mathbb{R}:\sum_{i\in I}p_{i}\tilde{\rho}_{AB}^{i}\leq\lambda\cdot\sum_{i\in I}p_{i}\tilde{\rho}_{A}^{i}\otimes\tilde{\sigma}_{B}^{i}\right\}+\log|I|\\
&\stackrel{\mathrm{(i)}}{\geq}\log\min\left\{\mu\in\mathbb{R}:\sum_{i\in I}p_{i}\tilde{\rho}_{AB}^{i}\leq\mu\cdot\sum_{i\in I}p_{i}\tilde{\rho}_{A}^{i}\otimes\sum_{j\in I}\tilde{\sigma}_{B}^{j}\right\}+\log|I|\\
&=\log\min\left\{\mu\in\mathbb{R}:\sum_{i\in I}p_{i}\tilde{\rho}_{AB}^{i}\leq\mu\cdot\sum_{i\in I}p_{i}\tilde{\rho}_{A}^{i}\otimes\sum_{j\in I}\frac{1}{l}\cdot\tilde{\sigma}_{B}^{j}\right\}\ ,
\end{align*}
where step (i) holds because $\sum_{i\in I}p_{i}\tilde{\rho}_{A}^{i}\otimes\sum_{j\in I}\tilde{\sigma}_{B}^{j}\geq\sum_{i\in I}p_{i}\tilde{\rho}_{A}^{i}\otimes\tilde{\sigma}_{B}^{i}$. Now set $\tilde{\sigma}_{B}=\sum_{j\in I}\frac{1}{l}\cdot\tilde{\sigma}_{B}^{j}$ and $\tilde{\rho}_{AB}=\sum_{i\in I}p_{i}\tilde{\rho}_{AB}^{i}$. By the convexity of the purified distance in its arguments (Lemma~\ref{333}) we obtain
\begin{align*}
\max_{i\in I}I_{\max}^{\eps}(A:B)_{\rho^{i}}+\log|I|&\geq\log\min\left\{\mu\in\mathbb{R}:\tilde{\rho}_{AB}\leq\mu\cdot\tilde{\rho}_{A}\otimes\tilde{\sigma}_{B}\right\}\\
&\geq\min_{\bar{\rho}_{AB}\in\cB^{\eps}(\rho_{AB})}\min_{\sigma_{B}\in\cS_{=}(\cH_{B})}\log\min\left\{\nu\in\mathbb{R}:\bar{\rho}_{AB}\leq\nu\cdot\bar{\rho}_{A}\otimes\sigma_{B}\right\}=I_{\max}^{\eps}(A:B)_{\rho}\ .
\end{align*}
\end{proof}

\subsection{Asymptotic Behavior}

The following is the \textit{Asymptotic Equipartition Property (AEP)} for smooth min- and max-entropy.

\begin{lemma}\cite[Theorem 9]{Tomamichel08}
Let $\eps>0$, $n\geq2\cdot(1-\eps^{2})$ and $\rho_{AB}\in\cS_{=}(\cH_{AB})$. Then
\begin{align}
&\frac{1}{n}H_{\min}^{\eps}(A|B)_{\rho^{\otimes n}}\geq H(A|B)_{\rho}-\frac{\eta(\eps)}{\sqrt{n}}\\
&\frac{1}{n}H_{\max}^{\eps}(A)_{\rho^{\otimes n}}\leq H(A)_{\rho}+\frac{\eta(\eps)}{\sqrt{n}}\ ,
\end{align}
where $\eta(\eps)=4\sqrt{1-2\cdot\log\eps}\cdot(2+\frac{1}{2}\cdot\log|A|)$.
\label{ja}
\end{lemma}

\begin{remark}\cite[Theorem 1]{Tomamichel08}
Let $\rho_{AB}\in\cS_{=}(\cH_{AB})$. Then
\begin{align}
&\lim_{\eps\rightarrow0}\lim_{n\rightarrow\infty}\frac{1}{n}H_{\min}^{\eps}(A|B)_{\rho^{\otimes n}}=H(A|B)_{\rho}\\
&\lim_{\eps\rightarrow0}\lim_{n\rightarrow\infty}\frac{1}{n}H_{\max}^{\eps}(A)_{\rho^{\otimes n}}=H(A)_{\rho}\ .
\end{align}
\label{newremark}
\end{remark}

\noindent
The Asymptotic Equipartition Property for the smooth max-information is as follows.

\begin{lemma}
Let $\eps>0$, $n\geq2\cdot(1-\eps^{2})$ and $\rho_{AB}\in\cS_{=}(\cH_{AB})$. Then
\begin{align}
\frac{1}{n}I_{\max}^{\eps}(A:B)_{\rho^{\otimes n}}\leq I(A:B)_{\rho}+\frac{\xi(\eps)}{\sqrt{n}}-\frac{2}{n}\cdot\log\frac{\eps^{2}}{24}\ ,
\end{align}
where $\xi(\eps)=8\sqrt{13-4\cdot\log\eps}\cdot(2+\frac{1}{2}\cdot\log|A|)$.
\label{haus}
\end{lemma}

\begin{proof}
Using the entropy measure upper bound for the smooth max-information (Lemma~\ref{hello}) together with the Asymptotic Equipartition Property for the smooth min- and max-entropies (Lemma~\ref{ja}) we obtain
\begin{align*}
\frac{1}{n}I_{\max}^{\eps}(A:B)_{\rho^{\otimes n}}&\leq\frac{1}{n}H_{\max}^{\eps^{2}/48}(A)_{\rho^{\otimes n}}-\frac{1}{n}H_{\min}^{\eps^{2}/48}(A|B)_{\rho^{\otimes n}}-\frac{2}{n}\cdot\log\frac{\eps^{2}}{24}\\
&\leq H(A)_{\rho}-H(A|B)_{\rho}-\frac{2}{n}\cdot\log\frac{\eps^{2}}{24}+2\cdot\frac{4}{\sqrt{n}}\sqrt{1-\log\left(\frac{\eps^{2}}{48}\right)^{2}}\cdot\log(2+\frac{1}{2}\cdot\log|A|)\\
&\leq I(A:B)_{\rho}+\frac{\xi(\eps)}{\sqrt{n}}-\frac{2}{n}\cdot\log\frac{\eps^{2}}{24}\ .
\end{align*}
\end{proof}

\begin{corollary}
Let $\rho_{AB}\in\cS_{=}(\cH_{AB})$. Then
\begin{align}
\lim_{\eps\rightarrow0}\lim_{n\rightarrow\infty}\frac{1}{n}I_{\max}^{\eps}(A:B)_{\rho^{\otimes n}}=I(A:B)_{\rho}\ .
\end{align}
\label{haus2}
\end{corollary}

\begin{proof}
By the Asymptotic Equipartition Property for the max-information (Lemma~\ref{haus}) we have
\begin{align*}
\lim_{\eps\rightarrow0}\lim_{n\rightarrow\infty}\frac{1}{n}I_{\max}^{\eps}(A:B)_{\rho^{\otimes n}}\leq I(A:B)_{\rho}+\lim_{\eps\rightarrow0}\lim_{n\rightarrow\infty}\left(\frac{\xi(\eps)}{\sqrt{n}}-\frac{2}{n}\cdot\log\frac{\eps^{2}}{24}\right)\ ,
\end{align*}
where $\xi(\eps)=8\sqrt{13-4\cdot\log\eps}\cdot(2+\frac{1}{2}\cdot\log|A|)$. Thus we obtain
\begin{align*}
\lim_{\eps\rightarrow0}\lim_{n\rightarrow\infty}\frac{1}{n}I_{\max}^{\eps}(A:B)_{\rho^{\otimes n}}\leq I(A:B)_{\rho}\ .
\end{align*}
To show the converse we need a Fannes-type inequality for the von Neumann entropy. We use Theorem 1 of~\cite{Audenaert07}:
Let $\rho$, $\sigma\in\cS_{=}(\cH)$ with $\rho\approx_{\eps}\sigma$. Then $|H(\rho)-H(\sigma)|\leq\eps\log(d-1)+H((\eps,1-\eps))$, where $d$ denotes the dimension of $\cH$.\\

\noindent
Because the max-relative entropy is always lower bounded by the relative von Neumann entropy (Lemma~\ref{relation}) we get
\begin{align*}
I_{\max}^{\eps}(A:B)_{\rho^{\otimes n}}=\min_{\bar{\rho}_{AB}^{n}\in\cB^{\eps}(\rho_{AB}^{\otimes n})}\min_{\sigma_{B}^{n}\in\cS_{=}(\cH^{\otimes n}_{B})}D_{\max}\left(\bar{\rho}_{AB}^{n}\|\bar{\rho}_{A}^{n}\otimes\sigma_{B}^{n}\right)\geq\min_{\bar{\rho}_{AB}^{n}\in\cB^{\eps}(\rho_{AB}^{\otimes n})}\min_{\sigma_{B}^{n}\in\cS_{=}(\cH^{\otimes n}_{B})}D(\bar{\rho}_{AB}^{n}\|\bar{\rho}_{A}^{n}\otimes\sigma_{B}^{n})\ .
\end{align*}
Noting that $D(\rho_{AB}\|\rho_{A}\otimes\sigma_{B})=D(\rho_{AB}\|\rho_{A}\otimes\rho_{B})+D(\rho_{B}\|\sigma_{B})$ and using the Fannes type inequality we then obtain
\begin{align*}
&\lim_{\eps\rightarrow0}\lim_{n\rightarrow\infty}\frac{1}{n}I_{\max}^{\eps}(A:B)_{\rho^{\otimes n}}
\geq\lim_{\eps\rightarrow0}\lim_{n\rightarrow\infty}\frac{1}{n}\min_{\bar{\rho}_{AB}^{n}\in\cB^{\eps}(\rho_{AB}^{\otimes n})}D(\bar{\rho}_{AB}^{n}\|\bar{\rho}_{A}^{n}\otimes\bar{\rho}_{B}^{n})\\
&\geq\lim_{\eps\rightarrow0}\lim_{n\rightarrow\infty}\frac{1}{n}\left\{D(\rho_{AB}^{\otimes n}\|\rho_{A}^{\otimes n}\otimes\rho_{B}^{\otimes n})-3\cdot H((\eps,1-\eps))-\eps\cdot\log\left[|A|^{n}|B|^{n}-1\right]-\eps\cdot\log\left[|A|^{n}-1\right]-\eps\cdot\log\left[|B|^{n}-1\right]\right\}\\
&\geq D(\rho_{AB}\|\rho_{A}\otimes\rho_{B})-\lim_{\eps\rightarrow0}\lim_{n\rightarrow\infty}\left\{\frac{3}{n}\cdot H((\eps,1-\eps))-\eps\cdot\left\{\log\left[|A||B|\right]+\log{|A|}+\log|B|\right\}\right\}=I(A:B)_{\rho}\ .
\end{align*}
\end{proof}

\subsection{Technical Lemmas}

\begin{lemma}~\cite[Lemma A.7]{Berta09}
Let $\rho\in\cS_{\leq}(\cH)$ and $\Pi\in\cP(\cH)$ such that $\Pi\leq\1$. Then
\begin{align}
P(\rho,\Pi\rho\Pi)\leq\tr[\rho]^{-1/2}\cdot\sqrt{\tr[\rho]^{2}-\tr[\Pi^{2}\rho]^{2}}\ .
\end{align}
\label{proj0}
\end{lemma}

\begin{lemma}
Let $\rho_{AB}\in\cS_{\leq}(\cH_{AB})$, $\sigma_{B},\omega_{B}\in\cS_{=}(\cH_{B})$ and $0\leq\Pi_{AB}\leq\1_{AB}$ with $\1_{A}\otimes\omega_{B}-\Pi_{AB}(\1_{A}\otimes\sigma_{B})\Pi_{AB}\geq0$. Then
\begin{align}
H_{\min}(A|B)_{\Pi\rho\Pi|\omega}\geq H_{\min}(A|B)_{\rho|\sigma}\ .
\end{align}
\label{kkk}
\end{lemma}

\begin{proof}
Set $H_{\min}(A|B)_{\rho|\sigma}=-\log\lambda$, i.e.~$\lambda$ is minimal such that $\lambda\cdot\1_{A}\otimes\sigma_{B}-\rho_{AB}\geq0$. Hence $\lambda\cdot\Pi_{AB}(\1_{A}\otimes\sigma_{B})\Pi_{AB}-\Pi_{AB}\rho_{AB}\Pi_{AB}\geq0$. Using $\1_{A}\otimes\omega_{B}-\Pi_{AB}(\1_{A}\otimes\sigma_{B})\Pi_{AB}\geq0$ we obtain
\begin{align*}
\lambda\cdot\1_{A}\otimes\omega_{B}-\Pi_{AB}\rho_{AB}\Pi_{AB}=\lambda\cdot(\1_{A}\otimes\omega_{B}-\Pi_{AB}(\1_{A}\otimes\sigma_{B})\Pi_{AB})+\lambda\cdot\Pi_{AB}(\1_{A}\otimes\sigma_{B})\Pi_{AB}-\Pi_{AB}\rho_{AB}\Pi_{AB}\geq0\ .
\end{align*}
The claim then follows by the definition of the min-entropy.
\end{proof}

\begin{lemma}
Let $\eps>0$ and $\rho_{A}\in\cS_{\leq}(\cH_{A})$. Then there exists $0\leq\Pi_{A}\leq\1_{A}$ such that $\rho_{A}\approx_{\eps}\Pi_{A}\rho_{A}\Pi_{A}$ and
\begin{align}
H_{\max}^{\eps^{2}/6}(A)_{\rho}\geq H_{R}(A)_{\Pi\rho\Pi}+2\cdot\log\frac{\eps^{2}}{6}\ .
\end{align}
\label{980}
\end{lemma}

\begin{proof}
By \cite[Lemma A.15]{Berta09} we have that for $\delta>0$ and $\rho_{A}\in\cS_{\leq}(\cH_{A})$, there exists $0\leq\Pi_{A}\leq\1_{A}$ such that $\tr\left[\left(\1_{A}-\Pi_{A}^{2}\right)\rho_{A}\right]\leq3\delta$ and $H_{\max}^{\delta}(A)_{\rho}\geq H_{R}(A)_{\Pi\rho\Pi}-2\cdot\log\frac{1}{\delta}$.\\

\noindent
Furthermore Lemma~\ref{proj0} shows that $\tr\left[\left(\1_{A}-\Pi_{A}^{2}\right)\rho_{A}\right]\leq3\delta$ implies $\rho_{A}\approx_{\sqrt{6\delta}}\Pi_{A}\rho_{A}\Pi_{A}$. For $\eps=\sqrt{6\delta}$ this concludes the proof.
\end{proof}

\section{Proof of the Decoupling Theorem} \label{app:decoupling}

Let $A'$ be of the same dimension as $A$. Denote by $F_{AA'}$ the the swap operator of $A\otimes A'$. Let $\Pi_{A}^{+}$ be the projector on the symmetric subspace of $A\otimes A$, and $\Pi_{A}^{-}$ the projector on the anti-symmetric subspace of $A\otimes A$. We need the following facts (see also~\cite{Horodecki06}).

\begin{itemize}
\item $F_{AA'RR'} = F_{AA'}\otimes F_{RR'}$
\item $\mathrm{rank}(\Pi_{A}^\pm)=|A|(|A|\pm1)/2$
\item $F^{2}=\1$
\item $\Pi_{A}^{\pm}=\frac{1}{2}(\1_{AA'}\pm F_{AA'})$
\item $\tr\left[F_{AA'}\right]=|A|$
\item $\tr\left[(\psi_{A}\otimes\phi_{A'})F_{AA'}\right]=\tr\left[\psi_{A}\phi_{A}\right]$
\item $\tr\left[(\psi_{AR}\otimes\psi_{A'R'})\cdot(\1_{AA'}\otimes F_{RR'})\right]=\tr\left[\psi_{R}^{2}\right]$
\end{itemize}

\begin{lemma} \label{lem:gen-decoup}
Let $A=A_{1}A_{2}$. Then
\begin{align}
\int_{U(A)}(U\otimes U)^\dagger(\1_{A_{2}A_{2'}}\otimes F_{A_{1}A_{1}'})(U\otimes U)dU\leq\frac{1}{|A_{1}|}\1_{AA'}+\frac{1}{|A_{2}|}F_{AA'}\ ,
\end{align}
where $dU$ is the Haar measure over the unitaries on system $A$.
\end{lemma}

\begin{proof}
For any $X$ that is Hermitian, it follows from Schur's lemma that
\begin{align*}
\int_{U(A)}(U^\dagger\otimes U^\dagger)X(U\otimes U)dU=a_{+}(X)\Pi_{A}^{+}+a_{-}(X)\Pi_{A}^{-}\ ,
\end{align*}
where $a_{\pm}(X)\cdot\mathrm{rank}(\Pi_{A}^{\pm})=\tr[X\Pi_{A}^{\pm}]$. Choosing $X=G=\1_{A_{2}A_{2}'}\otimes F_{A_{1}A_{1}'}$ we get
\begin{align*}
\tr\left[\Pi_{A}^{\pm}(\1_{A_{2}A_{2}'}\otimes F_{A_{1}A_{1}'})\right]&=\frac{1}{2}\tr\left[(\1_{AA'}\pm F_{AA'})(\1_{A_{2}A_{2}'}\otimes F_{A_{1}A_{1}'})\right]\\
&=\frac{1}{2}\tr\left[\1_{A_{2}A_{2}'}\otimes F_{A_{1}A_{1}'}\right]\pm\frac{1}{2}\tr\left[F_{AA'}(\1_{A_{2}A_{2}'}\otimes F_{A_{1}A_{1}'})\right]\\
&=\frac{1}{2}|A_{2}|^{2}\cdot|A_{1}|\pm\frac{1}{2}|A_{2}|\cdot|A_{1}|^{2}\ .
\end{align*}
Since $\mathrm{rank}(\Pi_{A}^{\pm})=\frac{1}{2}|A|(|A|\pm1)$ we get
\begin{align*}
a_{\pm}(G)=\frac{|A_{2}|^{2}|A_{1}|\pm|A_{2}||A_{1}|^{2}}{|A|(|A|\pm1)}=\frac{|A_{2}|\pm|A_{1}|}{|A|\pm1}\ .
\end{align*}
From
\begin{align*}
\frac{a_{+}(G)+a_{-}(G)}{2}=\frac{1}{2}\left(\frac{|A_{2}|+|A_{1}|}{|A|+1}+\frac{|A_{2}|-|A_{1}|}{|A|-1}\right)=\frac{|A_{2}||A|-|A_{1}|}{|A|^{2}-1}=\frac{1}{|A_{1}|}\cdot\frac{|A|^{2}-|A_{1}|^{2}}{|A|^{2}-1}\leq\frac{1}{|A_{1}|}
\end{align*}
and
\begin{align*}
\frac{a_{+}(G)-a_{-}(G)}{2}=\frac{1}{2}\left(\frac{|A_{2}|+|A_{1}|}{|A|+1}-\frac{|A_{2}|-|A_{1}|}{|A|-1}\right)=\frac{|A_{1}||A|-|A_{2}|}{|A|^{2}-1}=\frac{1}{|A_{2}|}\cdot\frac{|A|^{2}-|A_{2}|^{2}}{|A|^{2}-1}\leq\frac{1}{|A_{2}|}
\end{align*}
follows that
\begin{align*}
\int_{U(A)}(U^\dagger\otimes U^\dagger)G(U\otimes U)dU=a_{+}(G)\Pi_{A}^{+}+a_{-}(G)\Pi_{A}^{-} &
=\frac{a_{+}(G)+a_{-}(G)}{2}\1_{CC'}+\frac{a_{+}(G)-a_{-}(G)}{2}F_{CC'}\\
& \leq\frac{1}{|A_{1}|}\1_{AA'}+\frac{1}{|A_{2}|}F_{AA'}\ .
\end{align*}
\end{proof}

\begin{lemma} \label{lem:haar-help}
Let $\rho_{AR}\in\cP(\cH_{AR})$, $A=A_{1}A_{2}$ and $\sigma_{A_{1}R}(U)=\tr_{A_{2}}\left[(U\otimes\1_{R})\rho_{AR}(U\otimes\1_{R})^{\dagger}\right]$. Then
\begin{align}
\int_{U(A)}\tr\left[\sigma_{A_{1}R}(U)^2\right]dU\leq\frac{1}{|A_{1}|}\tr\left[\rho_{R}^2\right]+\frac{1}{|A_{2}|}\tr\left[\rho_{AR}^2\right]\ ,
\end{align}
where $dU$ is the Haar measure over the unitaries on system $A$.
\end{lemma}

\begin{proof}
Using Lemma~\ref{lem:gen-decoup} we have
\begin{align*}
& \int_{U(A)}\tr\left[\sigma_{A_{1}R}(U)^2\right]dU=\int_{U(A)}\tr\left[\left(\sigma_{A_{1}R}(U)\otimes\sigma_{A_{1}'R'}(U)\right)F_{A_{1}A_{1}'RR'}\right]dU\\ 
& =\int_{U(A)}\tr\left[(U\otimes U\otimes\1_{RR'})(\rho_{AR}\otimes\rho_{A'R'})(U\otimes U\otimes\1_{RR'})^{\dagger}(\1_{A_{2}A'_{2}}\otimes F_{A_{1}A_{1}'RR'})\right]dU\\
& =\tr\left[(\rho_{AR}\otimes\rho_{A'R'})\left(\int_{U(A)}(U\otimes U)^\dagger(\1_{A_{2}A'_{2}}\otimes F_{A_{1}A_{1}'})(U\otimes U)dU\otimes F_{RR'}\right)\right]\\
& \leq\tr\left[(\rho_{AR}\otimes\rho_{A'R'})\left(\left(\frac{1}{|A_{1}|}\1_{AA'}+\frac{1}{|A_{2}|}F_{AA'}\right)\otimes F_{RR'}\right)\right]=\frac{1}{|A_{1}|}\tr\left[\rho_{R}^2\right]+\frac{1}{|A_{2}|}\tr\left[\rho_{AR}^2\right]\ .
\end{align*}
\end{proof}

\begin{lemma}\cite[Lemma 5.1.3]{Ren05}
\label{lem:1}
Let $S$ be a Hermitian operator on $\cH$ and $\xi\in\cP(\cH)$. Then\footnote{The Hilbert-Schmidt norm is defined as $\|\Gamma\|_{2}=\sqrt{\tr\left[\Gamma^{\dagger}\Gamma\right]}$.}
$$\|S\|_{1}\leq\sqrt{\tr(\xi)}\left\|\xi^{-1/4}S\xi^{-1/4}\right\|_{2}\ .$$
\end{lemma}

\begin{proof}[Proof of Theorem~\ref{thm:decoupling}]
We show
\begin{align}
\log|A_{1}|\leq\frac{\log|A|+H_{C}(A|R)_{\rho}}{2}-\log\frac{1}{\eps}\ ,
\label{eq:decouplingcondition2}
\end{align}
which is sufficient because the conditional min-entropy is upper bounded by the quantum conditional collision entropy (Lemma~\ref{collision}). Using Lemma~\ref{lem:1} with $\xi=\1_{A_1}\otimes\omega_R$ for $\omega_{R}\in\cS_{=}(\cH_{R})$, it suffices to show that
\begin{align*}
\int_{U(A)}\left\|\left(\1_{A_{1}}\otimes\omega_{R}^{-1/4}\right)\left(\sigma_{A_{1}R}(U)-\tau_{A_{1}}\otimes\rho_{R}\right)\left(\1_{A_{1}}\otimes\omega_{R}^{-1/4}\right)\right\|^{2}_{2} dU\leq\frac{\eps^2}{|A_{1}|}\ . 
\end{align*}
We have
$$\left(\1_{A_{1}}\otimes\omega_R^{-1/4}\right)\sigma_{A_{1}R}(U)\left(\1_{A_{1}}\otimes\omega_{R}^{-1/4}\right)=\tr_{A_{2}}\left[(U\otimes\1_{R})\left(\1_{A}\otimes\omega_{R}^{-1/4}\right)\rho_{AR}\left(\1_{A}\otimes\omega_{R}^{-1/4}\right)(U\otimes\1_{R})^\dagger\right]\ .$$
Let $\tilde{\rho}_{AR}=\left(\1_{A}\otimes\omega_{R}^{-1/4}\right)\rho_{AR}\left(\1_{A}\otimes\omega_{R}^{-1/4}\right)$ and $\tilde{\sigma}_{A_{1}R}(U)=\tr_{A_{2}}\left[(U\otimes\1_R)\tilde{\rho}_{AR}(U\otimes\1_{R})^{\dagger}\right]$. Our inequality can then be rewritten as
\begin{align*}
\int_{U(A)}\left\|\tilde{\sigma}_{A_{1}R}(U)-\tau_{A_{1}}\otimes\tilde{\rho}_{R}\right\|^{2}_{2}dU\leq\frac{\eps^2}{|A_{1}|}\ .
\end{align*}
Using $\tau_{A_{1}}\otimes\tilde{\rho}_{R}=\int_{U(A)}\tilde{\sigma}_{A_{1}R}(U)dU$ we get
\begin{align*}
&\int_{U(A)}\left\|\tilde{\sigma}_{A_{1}R}(U)-\tau_{A_{1}}\otimes\tilde{\rho}_{R}\right\|^{2}_{2}dU=\int_{U(A)}\tr\left[\big(\tilde{\sigma}_{A_{1}R}(U)-\tau_{A_{1}}\otimes\tilde{\rho}_{R}\big)^{2}\right]dU\\
&=\int_{U(A)}\left\{\tr\left[\tilde{\sigma}_{A_{1}R}(U)^{2}\right]-\tr\left[\tilde{\sigma}_{A_{1}R}(U)(\tau_{A_{1}}\otimes\tilde{\rho}_{R})\right]-\tr\left[(\tau_{A_{1}}\otimes\tilde{\rho}_{R})\tilde{\sigma}_{A_{1}R}(U)\right]+\tr\left[(\tau_{A_{1}}\otimes\tilde{\rho}_{R})^{2}\right]\right\}dU\\
&=\int_{U(A)}\tr\left[\tilde{\sigma}_{A_{1}R}(U)^{2}\right]dU-\tr\left[(\tau_{A_{1}}\otimes\tilde{\rho}_{R})^2\right]=\int_{U(A)}\tr\left[\tilde{\sigma}_{A_{1}R}(U)^{2}\right]dU-\frac{1}{|A_{1}|}\cdot\tr\left[\tilde{\rho}_{R}^{2}\right]\\
&\stackrel{\mathrm{(i)}}{\leq}\frac{1}{|A_{2}|}\cdot\tr\left[\tilde{\rho}_{AR}^{2}\right]\stackrel{\mathrm{(ii)}}{\leq}\frac{\eps^{2}}{|A_{1}|}\ ,
\end{align*}
where (i) follows from Lemma~\ref{lem:haar-help} and (ii) follows from~\eqref{eq:decouplingcondition2} and the definition of $H_{C}(A|R)_{\rho}$.
\end{proof}

\section{The Post-Selection Technique}

We use a norm on the set of CPTP maps which essentially measures the probability by which two such mappings can be distinguished. The norm is known as diamond norm in quantum information theory~\cite{Kitaev97}. Here, we present it in a formulation which highlights that it is dual to the well-known completely bounded (cb) norm~\cite{Paulsen}. 

\begin{definition}[Diamond norm]
Let $\cE_{A}:\cL(\cH_{A})\mapsto\cL(\cH_{B})$ be a linear map. The \emph{diamond norm} of $\cE_{A}$ is defined as
\begin{align}
\|\cE_{A}\|_{\diamond}=\sup_{k\in\mathbb{N}}\|\cE_{A}\otimes\cI_{k}\|_{1}\ ,
\end{align}
where $\|\cF\|_{1}=\sup_{\sigma\in\cS_{\leq}(\cH)}\|\cF(\sigma)\|_{1}$ and $\cI_{k}$ denotes the identity map on states of a $k$-dimensional quantum system.
\label{kitaev}
\end{definition}

\begin{proposition}\cite{Kitaev97, Paulsen}
The supremum in Definition~\ref{kitaev} is reached for $k=|A|$. Furthermore the diamond norm defines a norm on the set of CPTP maps.
\label{feeling}
\end{proposition}

\noindent
Two CPTP maps $\cE$ and $\cF$ are called $\eps$-close if they are $\eps$-close in the metric induced by the diamond norm.

\noindent
\begin{definition}[De Finetti states]
Let $\sigma\in\cS_{=}(\cH)$ and $\mu(.)$ be a probability measure on $\cS_{=}(\cH)$. Then
\begin{align}
\zeta^{n}=\int\sigma^{\otimes n}\mu(\sigma)\in\cS_{=}(\cH^{\otimes n})
\end{align}
is called \emph{de Finetti state}.
\label{definetti}
\end{definition}

\noindent
The following proposition lies at the heart of the \textit{Post-Selection Technique}.

\begin{proposition}\cite{ChristKoenRennerPostSelect}
Let $\eps>0$ and $\cE^{n}_{A}$ and $\cF^{n}_{A}$ be CPTP maps from $\cL(\cH_{A}^{\otimes n})$ to $\cL(\cH_{B})$. If there exists a CPTP map $K_{\pi}$ for any permutation $\pi$ such that $(\cE^{n}_{A}-\cF^{n}_{A})\circ\pi=K_{\pi}\circ(\cE^{n}_{A}-\cF^{n}_{A})$, then $\cE^{n}_{A}$ and $\cF^{n}_{A}$ are $\eps$-close whenever
\begin{align}
\left\|((\cE^{n}_{A}-\cF^{n}_{A})\otimes\cI_{RR'})(\zeta^{n}_{ARR'})\right\|_{1}\leq\eps(n+1)^{-(|A|^{2}-1)}\ ,
\end{align}
where $\zeta^{n}_{ARR'}$ is a purification of the de Finetti state $\zeta_{AR}^{n}=\int\sigma_{AR}^{\otimes n}d(\sigma_{AR})$ with $\sigma_{AR}=\proj{\sigma}_{AR}\in\cS_{=}(\cH_{A}\otimes\cH_{R})$, $\cH_{A}\cong\cH_{R}$ and $d(.)$ the measure on the normalized pure states on $\cH_{A}\otimes\cH_{R}$ induced by the Haar measure on the unitary group acting on $\cH_{A}\otimes\cH_{R}$, normalized to $\int d(.)=1$. Furthermore we can assume without loss of generality that $|R'|\leq(n+1)^{|A|^{2}-1}$.
\label{posti}
\end{proposition}

\begin{theorem}\cite[Carath\'{e}odory]{Gruber93}
Let $d\in\mathbb{N}$ and $x$ be a point that lies in the convex hull of a set $P$ of points in $\mathbb{R}^{d}$. Then there exists a subset $P'$ of $P$ consisting of
$d+1$ or fewer point such that $x$ lies in the convex hull of $P'$.
\label{cara}
\end{theorem}

\begin{corollary}
Let $\zeta_{AR}^{n}=\int\sigma_{AR}^{\otimes n}d(\sigma_{AR})$ as in Proposition~\ref{posti}. Then $\zeta_{AR}^{n}=\sum_{i}p_{i}\left(\omega^{i}_{AR}\right)^{\otimes n}$ with $\omega^{i}_{AR}=\proj{\omega^{i}}_{AR}\in\cS_{=}(\cH_{AR})$, $i\in\{1,2,\ldots,(n+1)^{2|A||R|-2}\}$ and $p_{i}$ a probability distribution.
\label{mario}
\end{corollary}

\begin{proof}
We can think of $\zeta_{AR}^{n}$ as a normalized state on the symmetric subspace $\mathrm{Sym}^{n}(\cH_{AR})\subset\cH_{AR}^{\otimes n}$. The dimension of $\mathrm{Sym}^{n}(\cH_{AR})$ is bounded by $k=(n+1)^{|A||R|-1}$. Furthermore the normalized states on $\mathrm{Sym}^{n}(\cH_{AR})$ can be seen as living in an $m$-dimensional real vector space where $m=k-1+2\cdot\frac{k(k-1)}{2}=k^{2}-1$. Now define $S$ as the set of all $\xi_{AR}^{n}=\omega_{AR}^{\otimes n}$, where $\omega_{AR}=\proj{\omega}_{AR}\in\cS_{=}(\cH_{AR})$. Then $\zeta^{n}_{AR}$ lies in the convex hull of the set $S\subset\mathbb{R}^{k^{2}-1}$. Using Carath\'eodory's theorem (Theorem~\ref{cara}), we have that $\zeta_{AR}^{n}$ lies in the convex hull of a set $S'\subset S$ where $S'$ consists of at most $p=k^{2}-1+1=k^{2}$ points. Hence we can write $\zeta_{AR}^{n}$ as a convex combination of $p=(n+1)^{2|A||R|-2}$ extremal points in $S'$, i.e.~$\zeta_{AR}^{n}=\sum_{i}p_{i}(\omega^{i}_{AR})^{\otimes n}$, where $\omega^{i}_{AR}=\proj{\omega^{i}}_{AR}\in\cS_{=}(\cH_{AR})$, $i\in\{1,2,\ldots,(n+1)^{2|A||R|-2}\}$ and $p_{i}$ a probability distribution.
\end{proof}

\end{document}